\documentclass[envcountsame,11pt]{article}

\usepackage[top=3cm,bottom=3cm,left=3cm,right=3cm,headsep=10pt,a4paper]{geometry}

\usepackage{avant}
\usepackage{times}

\usepackage{microtype}
\usepackage[utf8]{inputenc}
\usepackage[T1]{fontenc}

\usepackage{amsmath}
\usepackage{amsfonts}
\usepackage{amssymb}

\usepackage{amsthm}

\usepackage{algorithmicx}
\usepackage{algpseudocode}

\usepackage{tikz,pgffor}
\usepackage{tkz-euclide}
\usetkzobj{all}
\usetikzlibrary{arrows}
\usetikzlibrary{shapes}
\usetikzlibrary{calc}
\usetikzlibrary{automata}
\usetikzlibrary{positioning}

\tikzstyle{label}=[shape=circle,draw,inner sep=0pt,minimum size=5mm]
\tikzstyle{tran}=[draw,->,>=stealth, rounded corners]

\usepackage{pgfplots}
\usepackage{enumerate,paralist}
\usepackage{pbox}

\usepackage{comment}
\usepackage{parskip}
\usepackage{multirow}
\usepackage{url}
\usepackage{graphicx}

\usepackage{graphics}
\usepackage{listings}
\usepackage{mathtools}
\usepackage{stmaryrd}
\usepackage{subcaption}

\usepackage{alltt}

\usepackage[pdftex]{hyperref}
\hypersetup{hidelinks}

\DeclareMathAlphabet{\mathpzc}{OT1}{pzc}{m}{it}

\newcommand{\cone}{cone}

\DeclareMathOperator{\eff}{\mathit{eff}}

\newcommand{\plen}{\text{L}}
\newcommand{\term}{\mathcal{L}}
\newcommand{\conf}{\mathit{C}}
\newcommand{\Inc}{\mathit{Inc}}

\newcommand{\A}{\ensuremath{\mathcal{A}}}
\newcommand{\calH}{\ensuremath{\mathcal{H}}}

\newcommand{\N}{\ensuremath{\mathbb{N}}}

\newcommand{\R}{\ensuremath{\mathbb{R}}}
\newcommand{\Q}{\ensuremath{\mathbb{Q}}}
\newcommand{\bigO}{\ensuremath{\mathcal{O}}}
\newcommand{\ce}[1]{\ensuremath{\left(#1 \right)}}

\newcommand{\size}[1]{|\!|#1|\!|}

\newcommand{\bx}{\ensuremath{\bold{x}}}
\newcommand{\by}{\ensuremath{\bold{y}}}
\newcommand{\bz}{\ensuremath{\bold{z}}}
\newcommand{\bu}{\ensuremath{\bold{u}}}
\newcommand{\bv}{\ensuremath{\bold{v}}}
\newcommand{\bw}{\ensuremath{\bold{w}}}
\newcommand{\bn}{\ensuremath{\bold{n}}}
\newcommand{\bc}{\ensuremath{\bold{c}}}

\newcommand{\EXPSPACE}{\textsf{EXPSPACE}}

\newcommand{\Effects}{\mathit{Inc}}
\newcommand{\VASS}{\A}

\newcommand{\Normals}{\mathit{Normals}}
\newcommand{\Norm}{\mathit{norm}}
\newcommand{\Decomp}{\mathit{Decomp}}

\newcommand{\Rset}{\mathbb{R}}

\date{}

\newtheorem{theorem}{Theorem}[section]
\newtheorem{lemma}[theorem]{Lemma}

\newtheorem{remark}[theorem]{Remark}

\newtheorem{definition}[theorem]{Definition}
\newtheorem{example}[theorem]{Example}

\title{\bf Efficient Algorithms for Checking Fast\\ Termination in VASS}

\author{Tom\'{a}\v{s} Br\'{a}zdil\\
    \small Faculty of Informatics, Masaryk University\\
    \texttt{\small brazdil@fi.muni.cz} \and
    Krishnendu Chatterjee\\
    \small IST Austria\\
    \texttt{\small krish.chat@gmail.com} \and
    Anton\'{\i}n Ku\v{c}era\\
    \small Faculty of Informatics, Masaryk University\\
    \texttt{\small kucera@fi.muni.cz} \and
    Petr Novotn\'{y}\\
    \small IST Austria\\
    \texttt{\small petr.novotny@ist.ac.at} \and
    Dominik Velan\\
    \small Faculty of Informatics, Masaryk University\\
     \texttt{\small xvelan1@fi.muni.cz}}

\begin{document}

\maketitle

\begin{abstract}
Vector Addition Systems with States (VASS) consists of a finite state space 
equipped with $d$ counters ($d$ is called the dimension), where in 
each transition every counter is incremented, decremented, or left unchanged.
VASS provide a fundamental model for analysis of concurrent processes, 
parametrized systems, and they are also used as abstract models for programs for bounds 
analysis. 
While termination is the basic liveness property that asks the qualitative question of
whether a given model always terminates or not, the more general quantitative question 
asks for bounds on the number of steps to termination.
In the realm of quantitative bounds a fundamental problem is to obtain asymptotic bounds
on termination time. 
Large asymptotic bounds such as exponential or higher already suggest that either 
there is some error in modeling, or the model is not useful in practice. 
Hence we focus on polynomial asymptotic bounds for VASS.
While some well-known approaches (e.g., lexicographic ranking functions) are neither 
sound nor complete with respect to polynomial bounds, other approaches only present
sound methods for upper bounds.
The existing approaches neither provide complete methods nor provide analysis of 
precise complexity bounds. 
In this work our main contributions are as follows:
First, for linear asymptotic bounds we present a sound and complete method for VASS, 
and moreover, our algorithm runs in polynomial time.
Second, we classify VASS according the normals of the vectors of the cycles.
We show that singularities in the normal are the key reason for asymptotic 
bounds such as exponential (even in three dimensions) and 
non-elementary (even in four dimensions) for VASS.
In absence of singularities, we show that the asymptotic complexity bound 
is always polynomial and of the form $\Theta(n^k)$, for some integer $k \leq d$.
We present an algorithm, with time complexity polynomial in the size of the VASS 
and exponential in dimension $d$, to compute the optimal $k$.
In other words, in absence of singularities, we present an efficient 
sound and complete method to obtain precise (not only upper, but matching 
upper and lower) asymptotic complexity bounds for VASS.
\end{abstract}

\section{Introduction}
\label{sec-intro}

\noindent{\em Static analysis for quantitative bounds.} 
Static analysis of programs reasons about programs without running them.
The most basic and important problem about \emph{liveness} properties studied 
in program analysis 
is the {\em termination} problem that given a program asks whether it 
always terminates.
The above problem seeks a {\em qualitative} or Boolean answer.  
However, given the recent interest in analysis of resource-constrained systems,
such as embedded systems, as well as for performance analysis,
it is vital to obtain quantitative performance characteristics.
In contrast to the qualitative termination, the quantitative termination 
problem asks to obtain bounds on the number of
steps to termination.
The quantitative problem, which is more challenging than the qualitative one, 
is of great interest in program analysis in various
domains, e.g.,
(a)~in applications domains such as hard real-time systems, worst-case 
guarantees are required; and (b)~the bounds are useful in early detection 
of egregious performance problems in large code bases~\cite{SPEED1}.

\smallskip\noindent{\em Approaches for quantitative bounds.} 
Given the importance of the quantitative termination problem 
significant research effort has been devoted, including important 
projects such as SPEED, COSTA~\cite{SPEED1,SPEED2,DBLP:journals/entcs/AlbertAGGPRRZ09}.
Some prominent approaches are the following:
\begin{itemize}
\item The worst-case execution time (WCET) analysis is an active field 
of research on its own (with primary focus on sequential loop-free code and hardware 
aspects)~\cite{DBLP:journals/tecs/WilhelmEEHTWBFHMMPPSS08}.

\item Advanced program-analysis techniques have also been developed for
asymptotic bounds, such as resource analysis using abstract interpretation and 
type systems~\cite{SPEED2,DBLP:journals/entcs/AlbertAGGPRRZ09,DBLP:conf/popl/JostHLH10,Hoffman1,Hoffman2},
e.g., linear invariant generation to obtain disjunctive and non-linear upper 
bounds~\cite{DBLP:conf/cav/ColonSS03}, or 
potential-based methods~\cite{Hoffman1,Hoffman2}.

\item Ranking functions based approach provides sound and complete approach 
for the qualitative termination problem, and for the quantitative problem it 
provides a sound approach to obtain asymptotic upper 
bounds~\cite{BG05,DBLP:conf/cav/BradleyMS05,DBLP:conf/tacas/ColonS01,DBLP:conf/vmcai/PodelskiR04,DBLP:conf/pods/SohnG91,DBLP:conf/vmcai/Cousot05,DBLP:journals/fcsc/YangZZX10,DBLP:journals/jossac/ShenWYZ13}.
\end{itemize}
In summary, the WCET approach does not consider asymptotic bounds, while the other approaches consider asymptotic bounds,
and present sound but not complete methods for upper bounds.

\smallskip\noindent{\em VASS and their modeling power.}
Vector Addition Systems (VASs)~\cite{KM69} or equivalently Petri Nets 
are fundamental models for analysis of parallel processes~\cite{EN94}.
Enriching VASs with an underlying finite-state transition structure gives 
rise to Vector Addition Systems with States (VASS).
Intuitively, a VASS consists of a finite set of control states and transitions
between the control states, and and a set of $d$ counters that hold non-negative 
integer values, where at every transition between the control states each counter 
is either incremented or decremented.
VASS are a fundamental model for concurrent processes~\cite{EN94}, and thus are 
often used for performing analysis of such 
processes~\cite{DKO13:conc-verification-vass,GM12:asynchronous-verification-TOPLAS,KKW10:dynamic-cutoff-detection,KKW12:coverability-proof-minim}.
Besides that, VASS have been used as models of parametrized 
systems~\cite{Bloem16},
as abstract models for programs for bounds and amortized analysis~\cite{SZV14}, 
as well as models of interactions between components of an API in 
component-based 
synthesis~\cite{FMWDR17:component-based-synthesis}.
Thus VASS provide a rich modeling framework for a wide class of problems
in program analysis.

\smallskip\noindent{\em Previous results for VASS.}
For a VASS, a {\em configuration} is a control state along with the values of 
counters. 
The termination problem for VASS can be defined as follows:
(a)~{\em counter termination} where the VASS terminates when one of the counters 
reaches value~0; 
(b)~{\em control-state termination} where given a set of terminating control 
states
the VASS terminates when one of the terminating states is reached.
The termination question for VASS, given an initial configuration, asks whether 
all paths from the configuration terminate.
The counter-termination problem is known to be \EXPSPACE-complete: the 
\EXPSPACE-hardness is shown in~\cite{Lipton:PN-Reachability,Esparza:PN} 
and the upper bound follows from~\cite{Yen92:Petri-Net-logic,AH11:Yen,FLLS11:EnGames}.

\smallskip\noindent{\em Asymptotic bounds analysis for VASS.}
While the qualitative termination problem has been studied extensively for VASS,
the problem of quantitative bounds for the termination problem has received much 
less attention. 
In general, even for VASS whose termination can be guaranteed, the number of 
steps
required to terminate can be non-elementary (tower of exponentials) in the 
magnitude of the initial configuration (i.e. in the maximal counter value 
appearing in the configuration). 
For practical purposes, bounds such as non-elementary or even exponential are too 
high as asymptotic complexity bounds, and the relevant complexity bounds are the 
polynomial ones. 
In this work we study the problem of computing asymptotic bounds for VASS, 
focusing on polynomial asymptotic bounds. 
Given a VASS and a configuration $c$, let $n_c$ denote the maximum value of the 
counters in $c$. 
If for all configurations $c$ all paths starting from $c$ terminate, then let 
$T_c$ 
denote the worst-case termination time from configuration $c$ 
(i.e., the maximum number of steps till termination among all paths starting from $c$).
The quantitative termination problem with {\em polynomial asymptotic bound} given a VASS 
and an integer $k$ asks whether the asymptotic worst-case termination time is at 
most a polynomial of degree $k$, i.e., whether there exists a constant $\alpha$ such that for 
all $c$ we have $T_c \leq \alpha \cdot n_c^k$. 
Note that with $k=1$ (resp., $k=2,3$) the problem asks for asymptotic linear 
(resp., quadratic, cubic) bounds on the worst-case termination time. 
The asymptotic bound problem is rather different from the qualitative 
termination problem for VASS, and even the decidability of this problem is not obvious.

\smallskip\noindent{\em Limitations of the previous approaches for polynomial bounds for VASS.}
In the analysis of asymptotic bounds there are three key aspects, namely,
(a)~soundness, (b)~completeness, and (c)~precise (or tight complexity) bounds. 
For asymptotic bounds, previous approaches (such as ranking functions, potential-based methods etc) 
are sound (but not complete) for upper bounds.
In other words, if the approaches obtain linear, or quadratic, or cubic bounds, 
then such bounds are guaranteed as asymptotic upper bounds (i.e., soundness is guaranteed),
however, even if the asymptotic bound is linear or quadratic, the approaches may 
fail to obtain any asymptotic upper bound (i.e., completeness is not guaranteed).
Another approach that has been considered for complexity analysis of programs 
are {\em lexicographic} ranking 
functions~\cite{ADFG10:lexicographic-ranking-flowcharts}.
We show that with respect to polynomial bounds lexicographic ranking functions 
are not sound, i.e., there exists VASS  for which lexicographic ranking 
function exists but the asymptotic complexity is exponential (see 
Example~\ref{ex:lex}).
Finally, none of the existing approaches are applicable for tight complexity bounds, 
i.e., the approaches consider $O(\cdot)$ bounds and are not applicable 
for $\Theta(\cdot)$ bounds.
In summary, previous approaches do not provide sound and complete method for polynomial 
asymptotic complexity of VASS; and no approach provide techniques for precise complexity analysis.

\smallskip\noindent{\em Our contributions.} 
Our main contributions are related to the complexity of the quantitative termination 
with polynomial asymptotic bounds for VASS and our results are applicable 
to counter termination.

\begin{compactenum}

\item We start with the important special case of linear asymptotic bounds. 
We present the first sound and complete algorithm that can decide linear asymptotic 
bounds for all VASS.
Moreover, our algorithm is an efficient one that has polynomial time complexity. This contrast sharply with \EXPSPACE-hardness of the qualitative termination problem and shows that deciding \emph{fast (linear) termination}, which seems even more relevant for practical purposes, is computationally easier than deciding qualitative termination.

\item Next, we turn our attention to polynomial asymptotic bounds. For simplicity, we restrict ourselves to VASS where the underlying finite-state transition structure is strongly connected (see Section~\ref{sec-concl} for more comments). Given such a VASS $\VASS$, for every short\footnote{A cycle $C$ is short if its length is bounded by the number of control states of a given VASS.} 
cycle $C$ of the $\VASS$, the effect of 
executing the short cycle once can be represented as a $d$-dimensional vector, an analogue of loop
summary (ignoring any nested sub-loops) for classical programs. Let $\Inc$ 
denote the set of all \emph{increments}, i.e., short cycle effects 
in $\VASS$. We investigate the geometric properties of $\Effects$ to derive 
complexity bounds on $\VASS$. The property playing a key role is whether all 
cycle effects in $\Effects$ lie on one side of some hyperplane in $\Rset^{d}$. 
Formally, each hyperplane is uniquely determined by its normal vector 
$\bn$ (a vector perpendicular to the hyperplane), and a hyperplane defined by 
$\bn$ \emph{covers} a vector 
effects $\bv$ if $\bv \cdot \bn \leq 0$, where ``$\cdot$'' is the dot product 
of vectors. Geometrically, the hyperplane defined by $\bn$ splits the whole 
$d$-dimensional space into two halves such that the normal $\bn$ points into 
one of the halves, and its negative $-\bn$ points into the other half. The 
hyperplane then ``covers'' vector 
$\bv$ if $\bv$ points into the same half as the vector $-\bn$.
We denote by $\Normals(\VASS)$ the set of all normals such that each 
$\bn\in\Normals(\VASS)$ covers all cycle effects in $\VASS$.
Depending on the properties of $\Normals(\VASS)$, we can distinguish the following cases:

\begin{itemize}
\item[(A)] {\em No normal:} if $\Normals(\VASS)=\emptyset$ 
(Fig.~\ref{fig:geom-A});
\item[(B)] {\em Negative normal:} if all $\bn\in\Normals(\VASS)$ have a 
negative component (Fig.~\ref{fig:geom-B});
\item[(C)] {\em Positive normal:} if there exists 
$\bn\in\Normals(\VASS)$ whose all components are positive 
(Fig.~\ref{fig:geom-C});
\item[(D)] {\em Singular normal:} if (C) does not hold, but there exists 
$\bn\in\Normals(\VASS)$ such that all components of $\bn$ are non-negative 
(in which some component of $\bn$ is zero, Fig.~\ref{fig:geom-D});
\end{itemize}
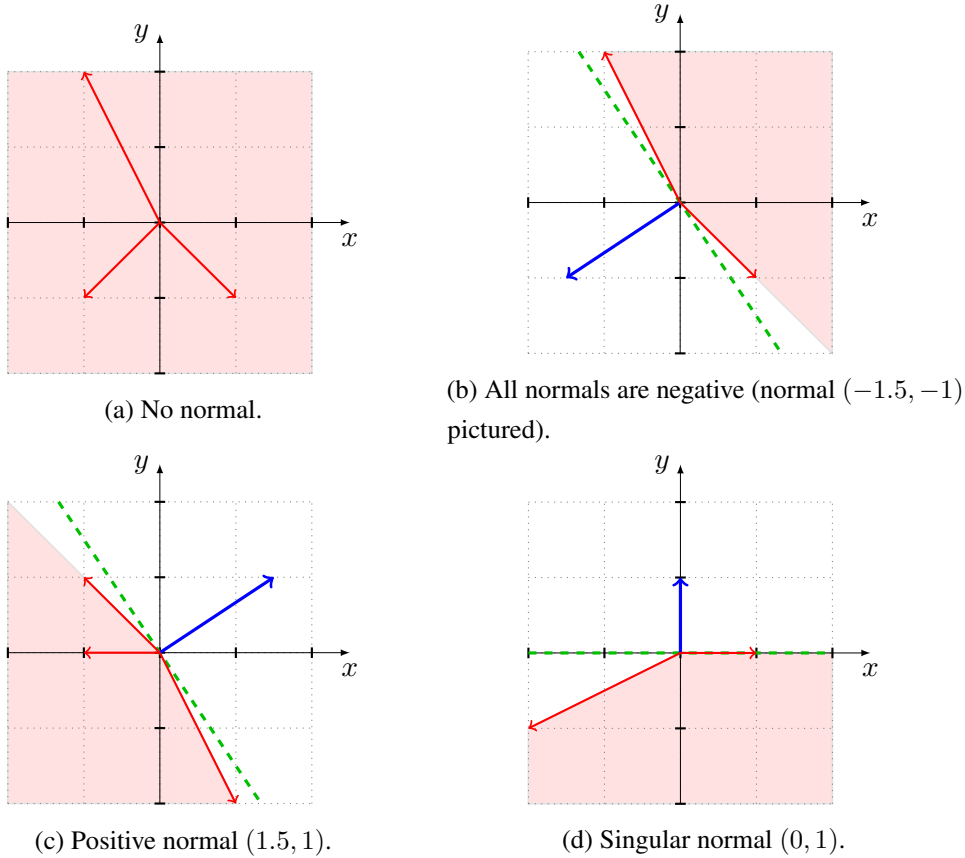
\begin{figure}[t]
		\centering
\begin{subfigure}{0.45\textwidth}
	\centering
\begin{tikzpicture}
	   \tkzInit[xmax=2,ymax=2,xmin=-2,ymin=-2]
   \tkzDefPoint(2,2){A}
   \tkzDefPoint(2,-2){B}
   \tkzDefPoint(-2,-2){C}
   \tkzDefPoint(-2,2){D}
   \tkzDrawPolygon[fill=red,opacity=.12](A,B,C,D)
   \begin{scope}[dotted]
   	\tkzGrid
   \end{scope}
   \begin{scope}[very thin]
   	\tkzDrawX
   	\tkzDrawY
   \end{scope}
   \draw[ thick,->, color = red] (0,0) -- (1,-1); 
   \draw[ thick,->, color = red] (0,0) -- (-1,2); 
   \draw[ thick,->, color = red] (0,0) -- (-1,-1); 
\end{tikzpicture}
\caption{No normal.}
\label{fig:geom-A}
\end{subfigure}
\begin{subfigure}{0.45\textwidth}
	\centering
\begin{tikzpicture}
\tkzInit[xmax=2,ymax=2,xmin=-2,ymin=-2]
\tkzDefPoint(2,2){A}
\tkzDefPoint(2,-2){B}
\tkzDefPoint(0,0){C}
\tkzDefPoint(-1,2){D}
\tkzDrawPolygon[fill=red,opacity=.12](A,B,C,D)
\begin{scope}[dotted]
\tkzGrid
\end{scope}
\begin{scope}[very thin]
\tkzDrawX
\tkzDrawY
\end{scope}
\draw[ very thick, dashed, color = green!75!black] (-4/3,2) -- (4/3,-2);
\draw[ very thick,->, color = blue] (0,0) -- (-1.5,-1);
\draw[  thick,->, color = red] (0,0) -- (1,-1); 
\draw[  thick,->, color = red] (0,0) -- (-1,2); 
\end{tikzpicture}
\caption{All normals are negative (normal $(-1.5,-1)$ pictured).}
\label{fig:geom-B}
\end{subfigure}

\begin{subfigure}{0.45\textwidth}
	\centering
\begin{tikzpicture}
\tkzInit[xmax=2,ymax=2,xmin=-2,ymin=-2]
\tkzDefPoint(-2,2){A}
\tkzDefPoint(0,0){B}
\tkzDefPoint(1,-2){C}
\tkzDefPoint(-2,-2){D}
\tkzDrawPolygon[fill=red,opacity=.12](A,B,C,D)
\begin{scope}[dotted]
\tkzGrid
\end{scope}
\begin{scope}[very thin]
\tkzDrawX
\tkzDrawY
\end{scope}
\draw[ very thick, dashed, color = green!75!black] (-4/3,2) -- (4/3,-2); 
\draw[ very thick,->, color = blue] (0,0) -- (1.5,1); 
\draw[  thick,->, color = red] (0,0) -- (-1,0); 
\draw[  thick,->, color = red] (0,0) -- (-1,1); 
\draw[  thick,->, color = red] (0,0) -- (1,-2); 
\end{tikzpicture}
\caption{Positive normal $(1.5,1)$.}
\label{fig:geom-C}
\end{subfigure}
\begin{subfigure}{0.45\textwidth}
	\centering
\begin{tikzpicture}
\tkzInit[xmax=2,ymax=2,xmin=-2,ymin=-2]
\tkzDefPoint(0,0){A}
\tkzDefPoint(2,0){B}
\tkzDefPoint(2,-2){C}
\tkzDefPoint(-2,-2){D}
\tkzDefPoint(-2,-1){E}
\tkzDrawPolygon[fill=red,opacity=.12](A,B,C,D,E)
\begin{scope}[dotted]
\tkzGrid
\end{scope}
\begin{scope}[very thin]
\tkzDrawX
\tkzDrawY
\end{scope}
\draw[ very thick, dashed, color = green!75!black] (-2,0) -- (2,0); 
\draw[  very thick,->, color = blue] (0,0) -- (0,1); 
\draw[  thick,->, color = red] (0,0) -- (1,0); 
\draw[  thick,->, color = red] (0,0) -- (-2,-1); 
\end{tikzpicture}
\caption{Singular normal $(0,1)$.}
\label{fig:geom-D}
\end{subfigure}
\caption{Classification of VASS into 4 sub-classes according to the geometric 
properties of vectors of cycle effects, pictured on 2D examples. Each figure 
pictures (as red arrows) 
vectors of simple cycle effects in some VASS (it is easy, for each 
figure, to construct a VASS whose simple cycle effects are exactly those 
pictured). The green dashed line, if present, represents the hyperplane (in 2D 
it is a line) 
covering the set of cycle effects. The thick blue arrow represents the normal 
defining the covering hyperplane. The pink shaded area represents the 
\emph{cone} generated by cycle effects (see Section~\ref{sec-VASS-term}). Intuitively, we 
seek hyperplanes that do not intersect the interior of the cone (but can touch 
its boundary).  }
\label{fig:geom-classification}
\end{figure}
First, we observe that given a VASS, we can decide to which of 
the above category it belongs, in time which is polynomial in the number of control states of a given VASS for every fixed dimension (i.e., the algorithm is exponential only in the dimension~$d$; see Section~\ref{sec-VASS-def} for more comments). Second, we also show that if a VASS belongs to one of the first two categories, 
then there exist configurations with non-terminating runs from them (see 
Theorem~\ref{thm-condAB}).
Hence asymptotic bounds are not applicable for the first two categories and 
we focus on the last two categories for polynomial asymptotic bounds.

\item For the positive normal category (C) we show that either there exist non-terminating 
runs or else the worst-case termination time is of the form $\Theta(n^k)$, where $k$ is an integer and $k\leq d$.
We show that given a VASS in this category, we can first decide whether all runs are terminating, and if yes, then we can compute the optimal asymptotic polynomial degree $k$ such that the worst-case termination time is $\Theta(n^k)$ (see Theorem~\ref{thm-poly-poly}). Again, this is achievable in time polynomial in the number of control states of a given VASS for every fixed dimension. 
In other words, for this class of VASS we present an efficient approach that is
sound, complete, and obtains precise polynomial complexity bounds.
To the best of our knowledge, no previous work presents a complete approach for
asymptotic complexity bounds for VASS, and the existing techniques only consider 
$O(\cdot)$ bounds, and not precise $\Theta(\cdot)$ bounds.

\item We show that singularities in the normal are the key reason for complex
asymptotic bounds in VASS. 
More precisely, for VASS falling into the singular normal category (D), in 
general 
the asymptotic bounds
are not polynomial, and we show that 
(a)~by slightly adapting the results of \cite{MayrMeyer:PN-containment}, it follows that termination complexity of a VASS $\A$ in category~(D) cannot be bounded by any primitive recursive function in the size of~$\A$;
(b)~even with three dimensions, the asymptotic bound is exponential in general (see Example~\ref{ex:exponential}), 
(c)~even with four dimensions, the asymptotic bound is non-elementary in general (see Example~\ref{ex:nonelem}).

\end{compactenum}

The main \emph{technical contribution} of this paper is a novel geometric approach, based on hyperplane separation techniques, for asymptotic time complexity analysis of VASS. Our methods are sound for arbitrary VASS and complete for a non-trivial subclass.

\section{Preliminaries}
\label{sec-prelim}

\subsection{Basic Notation}

We use $\N$, $\Q$, and $\R$ to denote the sets of non-negative integers, rational numbers, and real numbers. The subsets of all \emph{positive} elements of $\N$, $\Q$, and $\R$ are denoted by $\N^+$, $\Q^+$, and $\R^+$.  Further, we use $\N_\infty$ to denote the set $\N \cup \{\infty\}$ where $\infty$ is treated according to the standard conventions.
The cardinality of a given set $M$ is denoted by $|M|$. When no confusion arises, we also use $|c|$ to denote the absolute value of a given $c \in \R$. 

Given a function $f : \N \rightarrow \N$, we use  $\bigO(f(n))$ and $\Omega(f(n))$ to denote the sets of all $g : \N \rightarrow \N$ such that $g(n) \leq a \cdot f(n)$ and $g(n) \geq b \cdot f(n)$ for all sufficiently large $n \in \N$, where $a,b \in \R^+$ are some constants. If $h(n) \in \bigO(f(n))$ and $h(n) \in \Omega(f(n))$, we write $h(n) \in \Theta(f(n))$.

Let $d \geq 1$. The elements of $\R^d$ are denoted by bold letters such as $\bu,\bv,\bz,\ldots$. The $i$-th component of $\bv$ is denoted by $\bv(i)$, i.e., $\bv = (\bv(1),\ldots,\bv(d))$. For every $n \in \N$, we use $\vec{n}$ to denote the constant vector where all components are equal to~$n$.
The scalar product of $\bv,\bu \in \R^d$ is denoted by $\bv \cdot \bu$, i.e., 
$\bv\cdot \bu = \sum_{i=1}^d \bv(i)\cdot\bu(i)$. The other 
standard operations and relations on $\R$ such as 
$+$, $\leq$, or $<$ are extended to $\R^d$ in the component-wise way. In 
particular, $\bv$ is \emph{positive} if $\bv > \vec{0}$, i.e., all components of 
$\bv$ are positive. The norm of $\bv$ is defined by $\Norm(\bv) = \sqrt{\bv(1)^2 + \cdots+\bv(d)^2}$.

\paragraph{Half-spaces and Cones.}
An \emph{open half-space} of $\R^d$ determined by a normal vector $\bn \in 
\R^d$, where $\bn \neq \vec{0}$, is the set $\calH_{\bn}$ of all $\bx \in \R^d$ such that $\bx \cdot \bn < 0$. A \emph{closed half-space} $\hat{\calH}_{\bn}$ is defined in the same way 
but the above inequality is non-strict. Given a finite set of vectors $U 
\subseteq \R^d$, we use $\cone{}(U)$ to denote the set of all vectors of the form 
$\sum_{\bu \in U} c_{\bu} \bu$, where $c_{\bu}$ is a non-negative real constant 
for every $\bu \in U$.

\begin{example}
In Fig.~\ref{fig:geom-classification}, the cone, or more precisely its part 
that intersects the displayed area of $\Rset^2$,
generated by the cycle 
effects (i.e., by the ``red'' vectors) is the pink-shaded area. As for the half 
spaces, e.g., in Fig.~\ref{fig:geom-D}, the closed half-space defined by the 
normal 
vector $(0,1)$ is the set $\{(x,y)\mid y\leq 0\}$, while the open half-space 
determined by the same normal is the set $\{(x,y)\mid y< 0\}$. Intuitively, 
each normal vector $\bn$ determines a hyperplane (pictured by dashed lines in 
Fig.~\ref{fig:geom-classification}) that cuts $\Rset^d$ in two halves, and 
$\calH_{\bn}$ is the 
half which does not contain $\bn$: depending on whether we are interested in 
closed or open half-space, we include the separating hyperplane into 
$\calH_{\bn}$ or not, respectively.
\end{example}

\subsection{Syntax and semantics of VASS} 
\label{sec-VASS-def}

In this subsection we present a syntax of VASS, represented as finite state 
graphs with transitions labelled by vectors of counter changes.
 
\begin{definition}
\label{def-VASS}
Let $d \in \N^+$. A \emph{$d$-dimensional vector addition system with states (VASS)} is a pair $\A = \ce{Q,T}$, where $Q \neq \emptyset$ is a finite set of \emph{states} and  $T \subseteq Q \times \{-1,0,1\}^d \times Q$ is a set of \emph{transitions}. 
\end{definition}

\begin{example}
Fig.~\ref{fig:comp-ex} shows examples of three small 2-dimensional VASS. The VASS in Fig.~\ref{fig:comp-ex1} has two states $q_1,q_2$ and four  
transitions 
$(q_1,(-1,1),q_2)$, $(q_1,(0,0),q_2)$, $(q_2,(-1,0),q_1)$, $(q_2,(1,-1),q_2)$.
\end{example}

In some cases, we design algorithms where the time complexity is not polynomial in $\size{\A}$ (i.e., the size of $\A$), but polynomial in $|Q|$ and exponential just in~$d$. Then, we say that the running time is polynomial in $|Q|$ for a fixed~$d$.

\begin{figure}
\begin{subfigure}{0.45\textwidth}
\centering
	\begin{tikzpicture}[shorten >=1pt,node distance=2cm,on grid,auto]
	\node[state] (q_1)   {$q_1$};
	\node[state] (q_2) [right=of q_1] {$q_2$};
	\path[->]
	(q_1) edge [bend left]  node {(0,0)} (q_2)
	edge  [loop left] node {(-1,1)} ()
	(q_2) edge [bend left] node  {(-1,0)} (q_1)
	edge [loop right] node {(1,-1)} ();
	\end{tikzpicture}
\caption{Quadratic complexity.}
\label{fig:comp-ex1}
\end{subfigure}
\begin{subfigure}{0.45\textwidth}
\centering
	\begin{tikzpicture}[shorten >=1pt,node distance=2cm,on grid,auto]
	\node[state] (q_1)   {$q_1$};
	\node[state] (q_2) [right=of q_1] {$q_2$};
	\path[->]
	(q_1) edge [bend left]  node {(0,0)} (q_2)
	edge  [loop left] node {(-1,0)} ()
	(q_2) edge [bend left] node  {(-1,1)} (q_1)
	edge [loop right] node {(1,-1)} ();
	\end{tikzpicture}
\caption{Non-terminating VASS.}
\label{fig:comp-ex2}
\end{subfigure}
\begin{subfigure}{\textwidth}
	\centering
	\begin{tikzpicture}[shorten >=1pt,node distance=2cm,on grid,auto]
	\node[state] (q_1)   {$q_1$};
	\node[state] (q_2) [right=of q_1] {$q_2$};
	\path[->]
	(q_1) edge []  node {(0,0)} (q_2)
	edge  [loop left] node {(-1,1)} ()
	(q_2) edge [loop above] node  {(-1,0)} ()
	edge [loop right] node {(1,-1)} ();
	\end{tikzpicture}
\caption{Linear complexity.}
\label{fig:comp-ex3}
\end{subfigure}
\caption{An example of 2-dimensional VASS of varying complexity.}
\label{fig:comp-ex}
\end{figure}
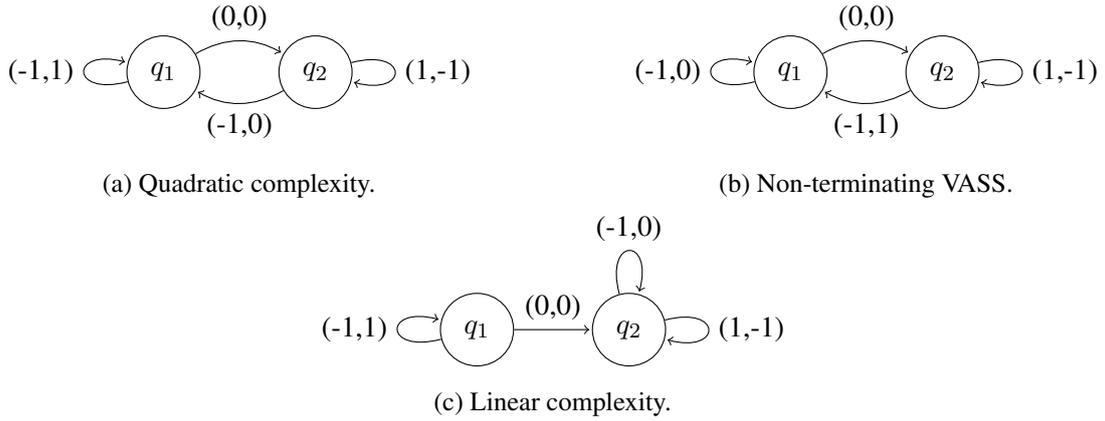

We use simple operational semantics for VASS based on the view of VASS as 
finite-state machines augmented with non-negative integer-valued counters.

A \emph{configuration} of $\A$ is a pair $p\bv$, where $p \in Q$ and $\bv \in 
\N^d$. The set of all configurations of $\A$ is denoted by $\conf(\A)$. The 
\emph{size} of  $p\bv \in \conf(\A)$ is defined as $\size{p\bv} = \max \{\bv(i) 
\mid 1\leq i \leq d\}$.

A \emph{finite path} in $\A$ of length~$n$ is a finite sequence $\pi$ of the form 
$p_0,\bu_1,p_1,\bu_2,p_2,\ldots,\bu_n,p_n$ where $n \geq 1$ and 
$(p_i,\bu_{i+1},p_{i+1}) \in T$ for all $0 \leq i < n$. If $p_0 = p_n$, then 
$\pi$ is a \emph{cycle}. A cycle is \emph{short} if $n \leq |Q|$. The \emph{effect} of $\pi$, denoted by $\eff(\pi)$, is the sum $\bu_1 + \cdots + \bu_n$. Given two finite paths $\alpha = p_0,\bu_1,\ldots,p_n$ and $\beta = q_0,\bv_1,\ldots,q_m$ such that $p_n = q_0$, we use $\alpha \odot \beta$ to denote the finite path $p_0,\bu_1,\ldots,p_n,\bv_1,\ldots,q_m$.

Let $\pi$ be a finite path in $\A$. A \emph{decomposition of $\pi$ into short\footnote{A standard technique for analysing paths in VASS are decompositions into \emph{simple} cycles, where all states except for $p_0$ and $p_n$ are pairwise different. The reason why we use short cycles instead of simple ones is clarified in Lemma~\ref{lem-inc-compute}.} cycles}, denoted by $\Decomp(\pi)$, is a finite list of short cycles (repetitions allowed) defined recursively as follows:
\begin{itemize}
	\item If $\pi$ does not contain any short cycle, then $\Decomp(\pi) = []$, where $[]$ is the empty list.
	\item If $\pi = \alpha \odot \gamma \odot \beta$ where $\gamma$ is the first short cycle occurring in $\pi$, then $\Decomp(\pi) = \mathtt{Concat}([\gamma],\Decomp(\alpha \odot \beta))$, where $\mathtt{Concat}$ is the list concatenation operator.  
\end{itemize}
Observe that if $\Decomp(\pi) = []$, then the length of $\pi$ is at most $|Q|-1$. Since the length of every short cycle is bounded by $|Q|$, the length of $\pi$ is asymptotically the \emph{same} as the number of elements in $\Decomp(\pi)$, assuming a fixed VASS~$\A$. 
  
Given a path $\pi = p_0,\bu_1,p_1,\bu_2,p_2,\ldots,\bu_n,p_n$ and an initial 
configuration $p_0 \bv_0$, the \emph{execution} of $\pi$ in $p_0 \bv_0$ is a 
finite sequence $p_0\bv_0,\ldots,p_n \bv_n$ of configurations where $\bv_i = 
\bv_0 + \bu_1 + 
\cdots + \bu_i$ for all $0 \leq i \leq n$. If $\bv_i \geq \vec{0}$ for all $0 
\leq i \leq n$, we say that  $\pi$ is \emph{executable} in $p_0\bv_0$.

\subsection{Termination Complexity of VASS}
\label{sec-VASS-term}

A \emph{zero-avoiding computation} of length $n$ initiated in a configuration 
$p\bv$ is a finite sequence of configurations $\alpha = q_0\bz_0, \ldots,q_n\bz_n$ 
initiated in $p\bv$ such that $\bz_i > \vec{0}$ for all $0 \leq i \leq n$, and 
for each $0\leq i<n$ there is a transition $(q_i,\bu,q_{i+1}) \in T$ where 
$\bz_{i+1} = \bz_i + \bu$. Every zero-avoiding computation $\alpha$ initiated in $q_0\bz_0$ determines a unique finite path $\pi_\alpha$ in $\A$ such that $\alpha$ is the execution of $\pi_\alpha$ in $q_0\bz_0$.

\begin{definition}
\label{def-VASS-compl}
Let $\A = \ce{Q,T}$ be a $d$-dimensional VASS. For every configuration $p\bv$ of $\A$, let $\plen(p\bv)$ be the least $\ell \in \N_\infty$ such that the length of every zero-avoiding finite computation initiated in $p\bv$ is bounded by $\ell$. The \emph{termination complexity} of $\A$ is a function $\term : \N \rightarrow \N$ defined by
\[
   \term(n) = \max \left\{\plen(p\bv) \mid p\bv \in \conf(\A) \mbox{ where } \size{p\bv} = n \right\}.
\]
If $\term(n) = \infty$ for some $n \in \N$, we say that $\A$ is \emph{non-terminating}, otherwise it is \emph{terminating}.
\end{definition}	
Observe that if $\A$ is non-terminating, then $\term(n) = \infty$ for all sufficiently large $n \in \N$. Further, if $\A$ is terminating, then $\term(n) \in \Omega(n)$. In particular, if $\term(n) \in \bigO(n)$, we also have $\term(n) \in \Theta(n)$.

Given a path $\pi = p_0,\bu_1,p_1,\bu_2,p_2,\ldots,\bu_n,p_n$ and an initial configuration $p_0 \bv_0$, the \emph{execution} of $\pi$ in $p_0 \bv_0$ is a finite sequence $p_0\bv_0,\ldots,p_n \bv_n$ where $\bv_i = \bv_0 + \bu_1 + \cdots + \bu_i$ for all $0 \leq i \leq n$. If $\bv_i \geq \vec{0}$ for all $0 \leq i \leq n$, we say that  $\pi$ is \emph{executable} in $p_0\bv_0$.

Let
\(
   \Inc = \{\eff(\pi) \mid \pi \text{ is a short cycle of } \A\}\,.
\)
The elements of $\Inc$ are called \emph{increments}. Note that if $\bu \in \Inc$, then $-|Q| \leq \bu(i) \leq |Q|$ for all $1 \leq i \leq d$. Hence, $|\Inc|$ is polynomial in $|Q|$, assuming $d$ is a fixed constant. Although the total number of all short cycles can be exponential in~$|Q|$, the set $\Inc$ is computable efficiently.\footnote{Note that Lemma~\ref{lem-inc-compute} would \emph{not} hold if we used simple cycles instead of short cycles, because the problem whether a given vector $\bv$ is an effect of a simple cycle is NP-complete, even if $d =1$ (NP-hardness follows, e.g., by a straightforward reduction of the Hamiltonian path problem).} 

\begin{lemma}
\label{lem-inc-compute}
	Let $\A= \ce{Q,T}$ be a $d$-dimensional VASS, and let $p \in Q$. The set $\Inc$ is computable in time $\bigO(\size{\A}^d)$, i.e., polynomial in $|Q|$ assuming $d$ is a fixed constant.
\end{lemma}
\begin{proof}
The set $\Inc$ is computable by the following standard algorithm: For all $q,q'\in Q$ and $1\leq k\leq n$, let $E^k_{q,q'}$ be the set of all effects of paths from $q$ to $q'$ of length exactly~$k$. Observe that 
\begin{itemize}
	\item $E^1_{q,q'}=\{\bu \mid (q, \bu, q')\in T\}$ for all $q,q'\in Q$;
	\item for every  $1<k\leq |Q|$, we have that
	$E^k_{q,q'}=\bigcup_{q''\in Q} \{\bv+\bu\mid \bv\in E^{k-1}_{q,q''}\text{ and } (q'',\bu,q')\in T\}$.
\end{itemize}
Obviously, $\Inc = \bigcup_{q\in Q} \bigcup_{k=1}^n E^k_{q,q}$, and the sets $E^k_{q,q'}$ for $k \leq |Q|$ are computable in time polynomial in $|Q|$, assuming $d$ is a fixed constant.
\end{proof}

A \emph{strongly connected component (SCC)} of $\A$ is maximal $R \subseteq Q$ such that for all $p,q \in R$ where $p \neq q$ there is a finite path from $p$ to $q$. Given a SCC $R$ of $Q$, we define the VASS $\A_R$ by restricting the set of control states to $R$ and the set of transitions to $T \cap (R \times \{-1,0,1\}^d \times R)$. We say that $\A$ is \emph{strongly connected} if $Q$ is a SCC of~$\A$.

\section{Linear Termination Time}
\label{sec-linear}

In this section, we give a complete and effective characterization of all VASS 
with \emph{linear} termination complexity. 

More precisely, we first provide a precise mathematical characterization of 
VASS with linear complexity: we show that if $\A$ is a $d$-dimensional VASS, 
then $\term(n) \in \bigO(n)$ iff there is an open half-space $\calH_{\bn}$ of 
$\R^d$ such that $\bn > \vec{0}$ and  $\Inc \subseteq \calH$. 

Next we show that the mathematical characterization of VASS of linear 
complexity is equivalent to the existence of a \emph{ranking function} of a 
special form for this VASS. We also show that existence of such a function for 
a given VASS $\A$ can 
be decided (and the function, if it exists, synthesized) in time polynomial in 
the size of $\A$. Hence, we obtain a sound and complete polynomial-time 
procedure for deciding whether a given VASS has linear termination complexity.

We start with the mathematical characterization.
Due to the next lemma, we can safely restrict ourselves to strongly connected VASS. A proof is trivial.

\begin{lemma}
	\label{lem-SCC}
	Let $d \in \N$, and let $\A = \ce{Q,T}$ be a $d$-dimensional VASS. Then $\term(n) \in  \bigO(n)$ iff $\term_R(n) \in \bigO(n)$ for every SCC $R$ of $Q$, where $\term_R(n)$ is the termination complexity of $\A_R$.
\end{lemma}

Now we show that if there is \emph{no} open half-space $\calH_{\bn}$ such that $\bn > \vec{0}$ and  $\Inc \subseteq \calH_{\bn}$, then there exist short cycles $\gamma_1,\ldots,\gamma_k$ and coefficients $b_1,\ldots,b_k \in \N^+$ such that the sum $\sum_{i=1}^k b_i \cdot \eff(\gamma_i)$ is non-negative. Note that this does \emph{not} yet mean that $\A$ is non-terminating---it may happen that the cycles $\pi_1,\ldots,\pi_k$ pass through disjoint subsets of control states and cannot be concatenated without including auxiliary finite paths decreasing the counters.

\begin{lemma}
\label{lem-no-halfspace}
   Let $\A = \ce{Q,T}$ be a $d$-dimensional VASS. If there is no open half-space $\calH_{\bn}$ of $\R^d$ such that $\bn > \vec{0}$ and  $\Inc \subseteq \calH_{\bn}$, then there exist $\bv_1,\ldots,\bv_k \in \Inc$ and $b_1,\ldots,b_k \in \N^+$ such that $k \geq 1$ and $\sum_{i=1}^k b_i\bv_i \geq \vec{0}$.
\end{lemma} 
\begin{proof}
   We distinguish two possibilities.
   \begin{itemize}
   	\item[(a)] There exists a closed half-space $\hat{\calH}_{\bn}$ of $\R^d$ 
   	such that $\bn > \vec{0}$ and  $\Inc \subseteq \hat{\calH}_{\bn}$.
   	\item[(b)] There is no closed half-space $\hat{\calH}_{\bn}$ of $\R^d$ such that $\bn > \vec{0}$ and  $\Inc \subseteq \hat{\calH}_{\bn}$.
   \end{itemize}
   \textbf{Case~(a).} We show that there exists $\bu \in \Inc$ such that $\bu \neq \vec{0}$ and $-\bu \in \cone{}(\Inc)$. Note that this immediately implies the claim of our lemma---since $-\bu \in \cone{}(\Inc)$, there are $\bv_1,\ldots,\bv_k \in \Inc$ and $c_1,\ldots,c_k \in \R^+$ such that $-\bu = \sum_{i=1}^k c_i \bv_i$. Since all elements of $\Inc$ are vectors of non-negative integers, we can safely assume $c_i \in \Q^+$ for all $1\leq i \leq k$. Let $b$ be the least common multiple of $c_1,\ldots,c_k$. Then $b\bu + (b\cdot c_1) \bv_1 + \cdots + (b\cdot c_k) \bv_k = \vec{0}$ and we are done.
   
   It remains to prove the existence of $\bu$. Let us fix a normal vector $\bn > \vec{0}$ such that $\Inc \subseteq \hat{\calH}_{\bn}$ and the set $\Inc_{\bn} = \{ \bv \in \Inc \mid \bv \cdot \bn < 0\}$ is \emph{maximal} (i.e., there is no $\bn' > \vec{0}$ satisfying $\Inc \subseteq \hat{\calH}_{\bn'}$ and $\Inc_{\bn} \subset \Inc_{\bn'}$). Further, we fix $\bu \in \Inc$ such that $\bu \cdot \bn = 0$. Note that such $\bu \in \Inc$ must exist, because otherwise $\Inc_{\bn} = \Inc$ which contradicts the assumption of our lemma. We show $-\bu \in \cone{}(\Inc)$. Suppose the converse. Then by Farkas' lemma there exists a separating hyperplane for $\cone{}(\Inc)$ and $-\bu$ with normal vector $\bn'$, i.e., $\bv \cdot \bn' \leq 0$ for all $\bv \in \Inc$ and $-\bu \cdot \bn' > 0$. Since $\bn > \vec{0}$, we can fix a sufficiently small $\varepsilon > 0$ such that the following conditions are satisfied:
   \begin{itemize}
   	 \item $\bn + \varepsilon \bn' > \vec{0}$,
   	 \item for all $\bv \in \Inc$ such that $\bv \cdot \bn < 0$ we have that $\bv \cdot (\bn + \varepsilon \bn') < \vec{0}$.
   \end{itemize}
   Let $\bw = \bn + \varepsilon \bn'$. Then $\bw > 0$, $\bv \cdot \bw < 0$ for all $\bv \in \Inc_{\bn}$, and $\bu \cdot \bw = \bu \cdot \bn + \varepsilon (\bu \cdot \bn') = \varepsilon (\bu \cdot \bn') < 0$. This contradicts the maximality of $\Inc_{\bn}$.
   \smallskip
   
   \noindent
   \textbf{Case~(b).} Let $B = \{\bv \in \R^d \mid \bv \geq \vec{0} \text{ and }  1  \leq \sum_{i=1}^d \bv(i)  \leq  2\}$. We prove $\cone{}(\Inc) \cap B \neq \emptyset$, which implies the claim of our lemma (there are $\bv_1,\ldots,\bv_k \in \Inc$ and $c_1,\ldots,c_k \in \Q^+$ such that $\sum_{i=1}^k c_i \bv_i \in B$). Suppose the converse, i.e., $\cone{}(\Inc) \cap B = \emptyset$. Since both $\cone{}(\Inc)$ and $B$ are closed and convex and $B$ is also compact, we can apply the ``strict'' variant of hyperplane separation theorem. Thus, we obtain a vector $\bn \in\R^d$ and a constant $c \in \R$ such that $\bx \cdot \bn < c$ and $\by \cdot \bn > c$ for all $\bx \in \cone{}(\Inc)$ and $\by \in B$. Since $\vec{0} \in \cone{}(\Inc)$, we have that $c > 0$. Further, $\bn \geq \vec{0}$ (to see this, realize that if $\bn(i) < 0$ for some $1 \leq i \leq d$, then $\by \cdot \bn < 0$ where $\by(i) = 1$ and $\by(j) = 0$ for all $j \neq i$; since $\by \in B$ and $c > 0$, we have a contradiction). Now we show $\bx \cdot \bn \leq 0$ for all $\bx \in \cone{}(\Inc)$, which contradicts the assumption of Case~(b). Suppose $\bx \cdot \bn > 0$ for some $\bx \in \cone{}(\Inc)$. Then $(m \cdot \bx) \cdot \bn > c$ for a sufficiently large $m \in \N$. Since $m \cdot \bx \in \cone{}(\Inc)$, we have a contradiction.  
\end{proof}

Now we give the promised characterization of all VASS with linear termination complexity. Our theorem also reveals that the VASS termination complexity is either linear or at least quadratic (for example, it cannot be that $\term(n) \in \Theta(n\log n)$).
\begin{theorem}
\label{thm-linear}
  Let $\A = \ce{Q,T}$ be a $d$-dimensional VASS. We have the following:
  \begin{itemize}
  	 \item[(a)] If there is an open half-space $\calH_{\bn}$ of $\R^d$ such that $\bn > \vec{0}$ and  $\Inc \subseteq \calH_{\bn}$, then $\term(n) \in \bigO(n)$.
  	 \item[(b)] If there is no open half-space $\calH_{\bn}$ of $\R^d$ such that $\bn > \vec{0}$ and  $\Inc \subseteq \calH_{\bn}$, then $\term(n) \in \Omega(n^2)$.
  \end{itemize}
\end{theorem}
\begin{proof}
  We start with~(a). Let $\calH_{\bn}$ be an open half-space of $\R^d$ such that $\bn > \vec{0}$ and  $\Inc \subseteq \calH_{\bn}$, and let $q\bu$ be a configuration of~$\A$. Note that 
  $\lceil \bn \cdot \bu \rceil \in \bigO(\size{q\bu)})$ because $\bn$ does not depend on $q\bu$. Let $\delta = \min_{\bv \in \Inc} |\bv \cdot \bn|$. Each short cycle decreases the scalar product of the normal $\bn$ and vector of counters by at least $\delta$. Therefore, for every zero-avoiding computation $\alpha$ initiated in $q\bu$ we have that $\Decomp(\alpha)$ contains at most 
  $\bigO(\size{q\bu})$ elements, so the length of $\alpha$ is $\bigO(\size{q\bu})$. 
  
  Now suppose there is no open half-space $\calH_{\bn}$ of $\R^d$ such that 
  $\bn > \vec{0}$ and  $\Inc \subseteq \calH_{\bn}$. We show that $\term(n) \in 
  \Omega(n^2)$, i.e., there exist $p \in Q$ and a constant $a \in \R^+$ such 
  that for all configurations $p \vec{n}$, where $n \in \N$ is sufficiently 
  large, there is a zero-avoiding computation initiated in $p\vec{n}$ whose 
  length is at least $a \cdot n^2$. Due to Lemma~\ref{lem-SCC}, we can safely 
  assume that $\A$ is strongly connected. By Lemma~\ref{thm-linear}, there are 
  $\bv_1,\ldots,\bv_k \in \Inc$ and $b_1,\ldots,b_k \in \N^+$ such that $k \geq 
  1$ and 
  \begin{equation}\label{eq:growingCycles}
\sum_{i=1}^k b_i\bv_i \geq \vec{0}.
\end{equation}
    As the individual short cycles with effects 
  $\bv_1,\ldots,\bv_k$ may proceed through disjoint sets of states,  
   they \emph{cannot} be trivially concatenated into one large 
  cycle with non-negative effect. Instead, we fix a control state $p \in Q$ and 
  a cycle $\pi$ initiated in $p$ visiting \emph{all} states of~$Q$ (here we 
  need that $\A$ is strongly connected). Further, for every $1 \leq i \leq k$ 
  we fix a short cycle $\gamma_i$ such that $\eff(\gamma_i) = \bv_i$. For every 
  $t \in \N$, let $\pi_t$ be a cycle obtained from $\pi$ by inserting precisely 
  $t\cdot b_i$ copies of every $\gamma_i$, where $1 \leq i \leq k$. Observe 
  that the inequality~(\ref{eq:growingCycles}) implies
  \begin{equation}\label{eq:iteratingCycles}
  \eff(\pi_t) = \eff(\pi) + t\cdot \sum_{i=1}^k b_i\bv_i \geq \eff(\pi)  \quad \text{ for every } t \in \N.
\end{equation}  
   For every configuration 
  $p\bu$, let $t(\bu)$ be the largest $t \in \N$ such that $\pi_t$ is 
  executable in $p\bu$ and results in a zero-avoiding computation. If 
  such a
  $t(\bu)$ does not exist, i.e. $\pi_t$ is executable in $p\bu$ for all $t\in 
  \N$, then $\A$ is non-terminating (since, e.g. $\bv_1$ must be non-negative 
  in such a case), and the proof is finished. Hence, we can assume that 
  $t(\bu)$ is well-defined for each $\bu$. Since the 
  cycles $\pi$ and 
  $\gamma_1,\ldots,\gamma_k$ have fixed 
  effects, there is $b \in \R^+$ such that for all configurations $p\bu$ where 
  all components of $\bu$ (and thus also $\size{p\bu}$) are above some 
  sufficiently large threshold $\xi$ 
we have that
  $t(\bu) \geq b \cdot 
  \size{p\bu}$, i.e. $t(\bu)$ grows asymptotically at least linearly with the 
  minimal component of $\bu$. Now, for every $n \in \N$, consider a 
  zero-avoiding computation 
  $\alpha(n)$ initiated in $p \vec{n}$ defined inductively as follows:  
  Initially, $\alpha(n)$ consists just of  $p\bu_0 = p \vec{n}$; if the 
  prefix of $\alpha(n)$ constructed so far ends in a configuration $p\bu_i$ 
  such that $t(\bu_i) \geq 1$ and $\bu_i \geq \vec{\xi}$ (an 
  event we call a \emph{successful hit}), then 
  the prefix is prolonged by executing the 
  cycle $\pi_{t(\bu_i)}$ (otherwise, the construction of $\alpha(n)$ stops). 
  Thus, $\alpha(n)$ is obtained from $p\vec{n}$ by applying the inductive rule 
  $I(n)$ times, where $I(n) \in \N_{\infty}$ is the number of successful hits 
  before the construction of $\alpha(n)$ stops.
  Denote by $p\bu_i$ the configuration visited by $\alpha(n)$ at $i$-th 
  successful hit. Now the inequality~(\ref{eq:iteratingCycles}) implies that  
  $\bu_i \geq \vec{n} + i \cdot \eff(\pi)$, so there exists a constant $e$ such 
  that $\size{p\bu_i}\geq n - i \cdot e$. In 
  particular 
  the 
  decrease of all components 
  of $\bu_i$ is at most linear in~$i$. This means that $I(n) \geq c \cdot n$ 
  for all sufficiently large $n \in \N$, where $c \in \R^+$ is a suitable 
  constant. But at the same time, upon each successful hit we have $\bu_i \geq 
  \vec{\xi}$, so length of the segment beginning with $i$-th successful hit and 
  ending with the $(i+1)$-th hit or with the last configuration of $\alpha(n)$ 
  is at least $b\cdot \size{p\bu_i}\geq b\cdot( n-i\cdot e)$.
  Hence, the length of $\alpha(n)$ is at least $\sum_{i=1}^{c\cdot n}b\cdot 
  (n-i\cdot e)$, i.e. quadratic.
\end{proof}

\begin{example}
Consider the VASS in Figure~\ref{fig:comp-ex3}. It consists of two strongly 
connected components, $\{q_1\}$ and $\{q_2\}$. In $\A_{\{ q_1\}}$ we have 
$\Inc=\{(-1,1)\}$. For $\bn=(1,\frac{1}{2})$ the open half-space 
$\hat{\calH}_{\bn}$ contains $\Inc$. Similarly, in $\A_{\{ q_2\}}$ we have 
$\Inc=\{(-1,0),(1,-1)\}$. For $\bn=(1,2)$ we again have that $\Inc$ is 
contained in open half-space $\hat{\calH}_{\bn}$. Hence, the VASS has linear 
termination complexity.

Now consider the VASS in Figure~\ref{fig:comp-ex1}. It is strongly connected 
and $\Inc=\{(-1,1),(-2,2),(1,-1),(2,-2),(-1,0)\}$. But there cannot be an open 
2-dimensional half-space (i.e. an open half-plane) containing two opposite 
vectors, e.g. $(-1,1)$ and $(1,-1)$, because for any line going through the 
origin such that $(-1,1)$ does not lie on the line it holds that $(1,-1)$ lies 
on the ``other side'' of the line than $(-1,1)$. Hence, the VASS in 
Figure~\ref{fig:comp-ex1} has at least quadratic termination complexity. The 
same 
argument applies to VASS in Figure~\ref{fig:comp-ex2}.
\end{example}

A straightforward way of checking the condition of Theorem~\ref{thm-linear} is 
to construct the corresponding linear constraints and check their feasibility 
by linear programming. This would yield an algorithm polynomial in $|\Inc|$, 
i.e., polynomial in $|Q|$ for every fixed dimension~$d$. Now we show that the 
condition can actually be checked in time \emph{polynomial in the size of 
$\A$}. We do this by showing that the mathematical condition stated in 
Theorem~\ref{thm-linear} is equivalent to the existence of a ranking function 
of a special type
for a given VASS. Formally, a \emph{weighted linear map}  for a VASS 
$\A=(Q,T)$ 
is defined by a vector of coefficients $\bc$ and by a set of weights 
$\{h_q\mid q\in Q \}$, one constant for each state of $\A$. The  
weighted linear map $\mu=(\bc,\{h_q\mid q\in Q\})$ defines a function (which 
we, slightly abusing the notation, also denote by $\mu$) assigning numbers to 
configurations as follows: $\mu(p\bv)=\bc\cdot\bv + h_p$. A weighted linear map 
$\mu$ is a \emph{weighted linear ranking (WLR) function} for $\A$ if $\bc\geq 
\vec{0}$ and 
there 
exists 
$\epsilon>0$ such that for each configuration $p\bv$ and each transition 
$(p,\bu,q)$ it holds $\mu(p\bv)\geq \mu(q(\bv+\bu)) + \epsilon$, which is 
equivalent, due to linearity, to 
\begin{equation}\label{eq:linearRankingFunction}
h_p-h_q\geq \bc\cdot \bu + \epsilon
\end{equation}
We show 
that weighted linear ranking functions provide a sound and complete method for 
proving linear termination complexity of VASS.

\begin{theorem}
\label{thm-linear-polytime} 
   Let $d \in \N$. The problem whether the termination complexity of a given 
   $d$-dimensional VASS is linear is solvable in time polynomial in the size 
   of~$\A$. More precisely, the termination complexity of a VASS $\A$ is linear 
   if and only if there exists a weighted linear ranking function for $\A$. 
   Moreover, 
   the existence of a weighted linear ranking function for $\A$ can be decided 
   in time 
   polynomial in $\size{\A}$.
\end{theorem}
\begin{proof}[Proof Sketch]
In the course of the proof we describe a polynomial time-algorithm for deciding 
whether given VASS has linear termination complexity. Once the 
algorithm is described, we will show that what it really does is checking the 
existence of a weighted linear ranking function for $\A$.
	
   Let us start by sketching the underlying intuition. Our goal is to decide, in polynomial time, whether there is an open half-space $\calH_{\bn}$ of $\R^d$ such that $\bn > \vec{0}$ and  $\Inc \subseteq \calH_{\bn}$. Note that this is equivalent to deciding whether there is an open half-space $\calH_{\bn}$ of $\R^d$ such that $\bn \geq \vec{0}$ and  $\Inc \subseteq \calH_{\bn}$ (since we demand $\calH_{\bn}$ to be open and the scalar product is continuous, $\bn\geq \vec{0}$ can be slightly tilted by adding a small $\vec{\delta}>0$ to obtain a positive vector with the~desired property).
    
   Given a vector $\bn\in \mathbb{R}^d$ and a configuration $q\bv$, we say that $\bv\cdot \bn$ is the \emph{$\bn$-value of $q\bv$}. Observe that if there is an open half-space $\calH_{\bn}$ such that $\bn \geq \vec{0}$ and  $\Inc \subseteq \calH_{\bn}$, then there is $\varepsilon>0$ such that the effect of every short cycle decreases the $\bn$-value of a configuration by at least $\varepsilon$. As every path can be decomposed into short cycles, every path steadily decreases the $\bn$-value of visited configurations. It follows that the mean change (per transition) of the $\bn$-value along an infinite path is bounded from above by $-\varepsilon/|Q|$. On the other hand, if the maximum mean change in $\bn$-values (over all infinite paths) is bounded from above by some negative constant, then every short cycle must decrease the $\bn$-value by at least this constant. So, it suffices to decide whether there is $\bn\geq \vec{0}$ such that for all infinite paths the mean change of the $\bn$-value is negative. Thus, we reduce our problem to the classical problem of maximizing the mean payoff over a decision process with rewards.
   Using standard results (see, e.g.,~\cite{Puterman:1994}), the latter problem polynomially reduces to the problem of solving a~linear program that is (essentially) equivalent to the inequality~(\ref{eq:linearRankingFunction}).
   Finally, the linear program can be solved in polynomial time using e.g.~\cite{Khachiyan:1979}.
\end{proof}

\begin{remark}
	The weighted linear ranking functions can be seen as a special case of 
	well-known linear ranking functions for linear-arithmetic 
	programs~\cite{DBLP:conf/tacas/ColonS01,DBLP:conf/vmcai/PodelskiR04}, in 
	particular state-based linear ranking functions, where a linear function of 
	program variables is assigned to each state of the control flow graph. WLR 
	ranking functions are indeed a special case, since the linear functions 
	assigned to various state are almost identical, and they differ only in the 
	constant coefficient $h_q$. Also, as the proof of the previous theorem 
	shows, 
	WLR functions in VASS can be computed directly by linear programming, 
	without 
	the need for any ``supporting invariants,'' since effect of a transition in 
	VASS is independent of the current values of the counters. Also, 
	well-foundedness (i.e. the fact that the function is bounded from below) is 
	guaranteed by the fact that $\bn\geq 0$ and counter values in VASS are 
	always 
	non-negative. It is a common knowledge that the existence of a state-based 
	linear ranking function for a 
	linear arithmetic program implies that the running time of the program is 
	linear in the initial valuation of program variables. Hence, our main 
	result 
	can be interpreted as proving that for VASS, state-based linear ranking 
	functions are both sound and~\emph{complete} for proving linear termination 
	complexity.
\end{remark}

\section{Polynomial termination time}
\label{sec-poly}

In this section we concentrate on VASS with polynomial termination complexity. For simplicity, we restrict ourselves to \emph{strongly connected VASS}. As we already indicated in Section~\ref{sec-intro}, our analysis proceeds by considering properties of \emph{normal vectors} perpendicular to hyperplanes covering the vectors of $\Inc$.  

\begin{definition}
	\label{def-normals}
	Let $\A = \ce{Q,T}$ be a $d$-dimensional VASS. The set $\Normals(\A)$ consists of all $\bn \in \R^d$ such that $\bn \neq \vec{0}$ and $\Inc \subseteq \hat{\calH}_{\bn}$ (i.e., $\bv \cdot \bn \leq 0$ for all $\bv \in \Inc$). 
\end{definition}

Let $\A$ be a strongly connected VASS. We distinguish four possibilities.
\begin{itemize}
   \item[(A)] $\Normals(\A) = \emptyset$.
   \item[(B)] $\Normals(\A) \neq \emptyset$ and all $\bn \in \Normals(\A)$ have a negative component.
   \item[(C)] There exists $\bn \in \Normals(\A)$ such that $\bn > \vec{0}$.
   \item[(D)] There exists $\bn \in \Normals(\A)$ such that $\bn \geq \vec{0}$ and~(C) does not hold.
\end{itemize}
Note that one can easily decide which of the four conditions holds by linear programming. Due to Lemma~\ref{lem-inc-compute}, the decision algorithm is polynomial in the number of control states of $\A$ (assuming $d$ is a fixed constant).  

We start by showing that a VASS satisfying (A) or (B) is non-terminating. A proof is given in Section~\ref{sec-condAB}.
\begin{theorem}
\label{thm-condAB}
	Let $\A = \ce{Q,T}$ be a $d$-dimensional strongly connected VASS such that (A)~or~(B) holds. Then $\A$ is non-terminating.
\end{theorem}

\subsection{VASS satisfying condition~(C)}

Assume $\A$ is a $d$-dimensional VASS satisfying~(C). We prove that if $\A$ is terminating, then $\term(n) \in \Theta(n^\ell)$ for some $\ell \in \{1,\ldots,d\}$. Further, there is a polynomial-time algorithm deciding whether $\A$ is terminating and computing the constant $\ell$ if it exists (assuming $d$ is a fixed constant).

A crucial tool for our analysis is a \emph{good normal}, introduced in the next definition.

\begin{definition}
\label{def-good-normal}
	Let $\A$ be a VASS. We say that a normal $\bn \in \Normals(\A)$ is \emph{good} if $\bn > \vec{0}$ and for every $\bv \in \cone{}(\Inc)$ we have that $-\bv \in \cone{}(\Inc)$ iff $\bv \cdot \bn = 0$.
\end{definition}

\begin{example}\label{ex:good_normal}
Consider the VASS of Fig.~\ref{fig:comp-ex1}. Here, a good normal is, e.g., the vector $\bn=(1, 1)$. Observe that the effects of both self-loops (on $q_1$ and $q_2$) belong to the hyperplane defined by $(1, 1)$. Note that these loops compensate each other's effects so long as we stay in the hyperplane (this is the defining property of the good normal). This allows us to zig-zag in the plane without "paying" with decrements in the $\bn$-value except when we need to switch between the loops (recall that the $\bn$-value of a configuration $q\bv$ is the product $\bv\cdot \bn$). This produces a path of quadratic length, which is asymptotically the worst case.
\end{example}

The next lemma says that a good normal always exists and it is computable efficiently. A proof can be found in Section~\ref{sec-good-normal}.

\begin{lemma}
\label{lem-good-normal}
    Let $\A$ be a $d$-dimensional VASS satisfying~(C). Then there exists a good normal computable in time polynomial in $|Q|$, assuming $d$ is a fixed constant. 
\end{lemma}

The next theorem is the key result of this section. It allows to reduce the analysis of termination complexity of a given VASS to the analysis of several smaller instances of the problem, which can be then solved recursively.

\begin{theorem}
\label{thm-A-decomp}
	Let $\A$ be a VASS satisfying~(C), and let $\bn \in \Normals(\A)$ be a good normal. Consider a VASS $\A^{\bn} = (Q,T_{\bn})$ where
	\[
	   T_{\bn} = \{t \in T \mid \text{ there is a short cycle $\gamma$ of $\A$ containing $t$ such that } \eff(\gamma) \cdot \bn = 0 \}\,.
	\] 
	Further, let $C_1,\ldots,C_k$ be all SCC of $\A^{\bn}$ with at least one transition. We have the following:
	\begin{itemize}
		\item[(1)] If $k=0$ (i.e., if there is no SCC of  $\A^{\bn}$ with at least one transition), then $\term_{\A}(n) \in \Theta(n)$.
		\item[(2)] If $k>0$, all $\A^{\bn}_{C_1},\ldots,\A^{\bn}_{C_k}$ are terminating, and the termination complexity of every $\A^{\bn}_{C_i}$ is $\Theta(f_i(n))$,
		then $\A$ is terminating and $\term_{\A}(n) \in \Theta(n \cdot \max[f_1,\ldots,f_k](n))$, where $\max[f_1,\ldots,f_k] : \N \rightarrow \N$ is a function defined by $\max[f_1,\ldots,f_k](n) = \max\{f_1(n),\ldots,f_k(n)\}$.
	\end{itemize}
\end{theorem}

To get some intuiting behind the proof of Theorem~\ref{thm-A-decomp}, consider the following example.

\begin{example}
Consider the VASS of Fig.~\ref{fig:comp-ex1}. As mentioned in Example~\ref{ex:good_normal}, there is a good normal $\bn = (1, 1)$, which gives $T_{\bn}=\{(-1, 1), (1, -1)\}$. Then Case~(2) of 
Theorem~\ref{thm-A-decomp} gives us two simpler VASS $\A^{\bn}_{C_1},\A^{\bn}_{C_2}$ where 
$\A^{\bn}_{C_1}$ has a single state $q_1$ and a single transition $(q_1, (-1, 1), q_1)$, and $\A^{\bn}_{C_2}$ has a single state $q_2$ and a single transition $(q_2, (1, -1), q_2)$. Observe that both $\A^{\bn}_{C_1}$ and $\A^{\bn}_{C_2}$ can now be considered individually, and both of them have linear complexity. Also, as mentioned in Example~\ref{ex:good_normal}, the good normal makes sure that the effect of the worst case behavior in $\A^{\bn}_{C_1}$ can be compensated by a path in $\A^{\bn}_{C_2}$, and vice versa. Moreover, following the worst case path in $\A^{\bn}_{C_1}$ and its compensation in $\A^{\bn}_{C_2}$ decreases the final $\bn$-value of configurations only by a constant (caused by the switch between $\A^{\bn}_{C_1}$ and $\A^{\bn}_{C_2}$). So, we can follow such ``almost compensating'' loop $\Omega(n)$ times, and obtain a path of quadratic length.

Note that in the general case the situation is more complicated since the compensating path may need to be composed using paths in several VASS of $\A^{\bn}_{C_1},\ldots,\A^{\bn}_{C_k}$. So, we need to be careful about the number of switches and about geometry of the compensating path.
\end{example}

\begin{proof}[Proof sketch for Theorem~\ref{thm-A-decomp}]
	Claim~(1) follows easily. It suffices to realize that if there is no SCC of $\A^{\bn}$ with at least one transition, then there is no $\bv \in \Inc$ satisfying $\bv \cdot \bn = 0$. Hence, $\bv \cdot \bn < 0$ for all $\bv \in \Inc$, and we can apply Theorem~\ref{thm-linear}. 
	
	Now we prove Claim~(2). Let $\alpha$ be a zero-avoiding computation of $\A$ initiated in a configuration $q\bu$. Since the last configuration $p\bv$ of $\alpha$ satisfies $\bv \geq \vec{0}$, we have that $\bv \cdot \bn \geq 0$. Hence, 
	\[
	  \bv \cdot \bn \quad = \quad (\bu + \eff(\pi_\alpha)) \cdot \bn \quad = \quad \bu \cdot \bn +  \eff(\pi_\alpha) \cdot \bn \quad \geq \quad 0 \,.
	\]
	Let $\Decomp(\pi_\alpha)$ by a decomposition of $\pi_\alpha$ into short cycles. For every short cycle $\gamma$ of $\A$ we have that $\eff(\gamma) \cdot \bn \leq 0$. Since $\pi_\alpha$ can contain at most $|Q|$ transitions which are not contained in any cycle, we have that $\bu \cdot \bn \leq \bv \cdot \bn + c$, where $c \in \N$ is some fixed constant. This means that $\size{p\bv}$ is $\bigO(\size{q\bu})$. Consequently, the same holds also for all \emph{intermediate} configurations visited by~$\alpha$. 
	
	A short cycle $\gamma$ of $\A$ such that $\eff(\gamma) \cdot \bn < 0$ is called \emph{$\bn$-decreasing}, otherwise it is \emph{$\bn$-neutral}. Clearly, the total number of $\bn$-decreasing short cycles in $\Decomp(\pi_\alpha)$ is $\bigO(\size{q\bu})$, because each of them decreases the scalar product with $\bn$ by a fixed constant bounded away from zero, and $\bu \cdot \bn$ is $\bigO(\size{q\bu})$.  This means that the total number of transitions in $\pi_\alpha$ which are \emph{not} in $T_\bn$ is $\bigO(\size{q\bu})$ (as we already noted, $\pi_\alpha$ can also contain transitions which are not contained in any short cycle, but their total number is bounded by $|Q|$). Let $\varrho$ be a subpath of $\pi_\alpha$ with maximal length containing only transitions of $T_{\bn}$. Note that $\varrho$ is a concatenation of at most $|Q|$ subpaths which contain transitions of the same SCC $C_i$ of $\A^{\bn}$. Each of these subpaths is initiated in a configuration of size $\bigO(\size{q\bu})$, and hence its length is $\bigO(f_i(\size{q\bu}))$. Hence, the length of $\varrho$ is $\bigO(\size{q\bu} \cdot\max[f_1,\ldots,f_k](\size{q\bu}))$.
	
	It remains to prove that $\term_{\A}(n) \in \Omega(n \cdot \max[f_1,\ldots,f_k](n))$. Let us fix some $i \leq k$. We prove that there exists a constant $a \in \R^+$  such that for all sufficiently large  $n$ there exists a zero-avoiding computation $\alpha_n$ of length at least $a \cdot n \cdot f_i(n)$ initiated in a configuration of size $n$. The construction of $\alpha_n$ is technically non-trivial, so we first explain the underlying idea informally.  A formal proof is given in Section~\ref{sec-A-decomp}.
	
	To achieve the length  $\Omega(n \cdot f_i(n))$, the computation $\alpha_n$ needs to execute $\Omega(n)$ paths of length $\Theta(f_i(n))$ ``borrowed'' from $\A^{\bn}_{C_i}$. The problem is that even after executing just one path $\pi$ of length $\Theta(f_i(n))$, some counters can have very small values, which prevents the executing of another path of length $\Theta(f_i(n))$. Therefore, we need to ``compensate'' the effect of $\pi$ and increase the counters. This is where we use the properties of a good normal. We can choose $\pi$ so that it forms a cycle in $\A^{\bn}_{C_i}$ (not necessarily a short one), and we prove that all cycles in $\A^{\bn}_{C_i}$ are $\bn$-neutral. From this we get $-\eff(\pi) \in \cone{}(\Inc)$, and hence the effect of  $\pi$ can be compensated by an appropriate combination of short cycles of $\A^{\bn}$. So, after executing $\pi$, we execute the corresponding ``compensating'' path, and this is repeated $\Omega(n)$ times. Note that we need to ensure that the compensating paths do not decrease the counters too much in intermediate configurations, and the compensation ends in a configuration which is sufficiently close to the original configuration where we started executing~$\pi$.
\end{proof}

Now we can formulate and prove the main result of this section. 

\begin{theorem}
\label{thm-poly-poly}
	Let $\A$ be a $d$-dimensional VASS satisfying~(C). The problem whether $\A$ is terminating is decidable in time polynomial in $|Q|$, assuming $d$ is a fixed constant. Further, if $\A$ is terminating, then $\term(n) \in \Theta(n^k)$, where $k \in \{1,\ldots,d\}$ is a constant computable in time polynomial in $|Q|$, assuming $d$ is a fixed constant.  
\end{theorem}

\begin{proof}
	For a given a $\A$, the algorithm starts by computing a good normal $\bn$ (see Lemma~\ref{lem-good-normal}) and constructing the VASS $\A^{\bn} = (Q,T_{\bn})$ of Theorem~\ref{thm-A-decomp}. Here, the set $T_\bn$ is computed as follows. Note that $\A$ can be seen as a directed multigraph where the nodes are the states and the edges correspond to transitions. To every transition $(q, \bu, q')$ we assign its weight $-\bu\cdot \bn$. Note that the multigraph does not contain any negative cycles (a negative cycle in the multigraph would induce a cycle in $\A$ increasing the $\bn$-value; however, such a cycle cannot exist with a good normal $\bn$). To decide whether a given transition $(q, \bu, q')$ belongs to $T_{\bn}$, it suffices to find a~path with the least accumulated weight from $q'$ to $q$ (which can be done using, e.g., Bellman-Ford algorithm~\cite{Shimbel:1951}) and check whether the accumulated weight is equal to $\bu\cdot \bn$. Hence, $T_\bn$ is computable in time polynomial in the size of $\A$ (for a given good normal $\bn$).
	
	Then, the algorithm proceeds by constructing the SCC $C_1,\ldots,C_k$ of $\A^{\bn}$. If $k=0$, then $\term(n) \in \Theta(n)$ (see Theorem~\ref{thm-A-decomp}~(1)). If $k=1$ and $\A^{\bn}_{C_1} = \A$, then $\A$ is non-terminating (this is a consequence of Theorem~\ref{thm-A-decomp}~(2); if $\A^{\bn}_{C_1}$ was terminating with termination complexity $\Theta(f_1(n))$, then by Theorem~\ref{thm-A-decomp}~(2), the termination complexity of $\A = \A^{\bn}_{C_1}$ is $\Theta(n \cdot f_1(n))$, which is impossible). Otherwise, the algorithm proceeds by analyzing $\A^{\bn}_{C_1},\ldots,\A^{\bn}_{C_k}$ recursively. If some of them is non-terminating, then $\A$ is also non-terminating. Otherwise, the termination complexity of $\A$ is derived from the termination complexity of $\A^{\bn}_{C_1},\ldots,\A^{\bn}_{C_k}$ as in Theorem~\ref{thm-A-decomp}~(2). Clearly, we obtain $\term(n) \in \Theta(n^k)$ for some $k \in \{1,\ldots, d\}$. It is easy to verify that the total the number of recursive calls is polynomial in the size of~$\A$.
\end{proof}

\subsection{VASS satisfying condition~(D)}

Condition~(D) is not sufficiently strong to guarantee polynomial termination 
time for terminating VASS. In fact, as $d$ increases, the termination 
complexity can grow \emph{very} fast. 
Even for $d=3$, one can easily construct a terminating VASS satisfying~(C) such 
that $\term(n) \in \Omega(2^n)$. 

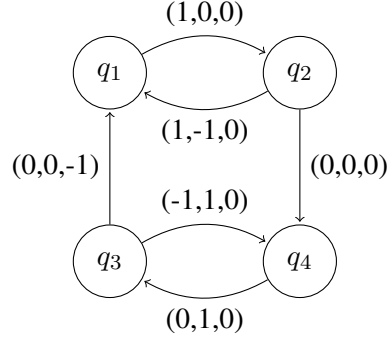
\begin{figure}
	\centering
	\begin{tikzpicture}[shorten >=1pt,node distance=2.5cm,on grid,auto]
	\node[state] (q_1)   {$q_1$};
	\node[state] (q_2) [right=of q_1] {$q_2$};
	\node[state] (q_3) [below=of q_1] {$q_3$};
	\node[state] (q_4) [right=of q_3] {$q_4$};
	\path[->]
	(q_1) edge [bend left] node {(1,0,0)} (q_2)
	(q_2) edge [bend left] node  {(1,-1,0)} (q_1)
	edge node {(0,0,0)} (q_4)
	(q_3) edge [bend left] node {(-1,1,0)} (q_4)
	edge node {(0,0,-1)} (q_1)
	(q_4) edge [bend left] node  {(0,1,0)} (q_3)
	;
	\end{tikzpicture}
	\caption{A 3-dimensional VASS satisfying condition (D) which has an exponential 
		termination complexity.}
	\label{fig:exponential}
\end{figure}

\begin{example}
	\label{ex:exponential}
	Consider the strongly connected 3-dimensional VASS $\A$ in 
	Fig.~\ref{fig:exponential}. Let $n\in\N$ be arbitrary. We construct a 
	zero-avoiding computation $\alpha(n)$ starting in $q_1\vec{n}$ whose length is 
	exponential in $n$. For better readability, denote by $x$, $y$, and $z$ the 
	variables representing the first, second, and third counter, respectively.
	
	The construction consist of iterating several phases. In Phase (a) we iterate 
	the 
	short cycle $q_1,(1,0,0),q_2,(1,-1,0),q_1$ as long as $y\geq 2$. Then we 
	perform the path $q_1,(1,0,0),q_2,(0,0,0),q_4$ to $q_4$. From there we continue 
	with Phase (b), where we iterate the short cycle $q_4,(0,1,0),q_3,(-1,1,0),q_4$ 
	as long as $x\geq 2$. After this we perform the path 
	$q_4,(0,1,0),q_3,(0,0,-1),q_1$ to $q_1$. There we again switch to Phase (a), 
	repeating the process until one of the counters hits zero.
	
	One can straightforwardly check that the total effect of performing Phase (a) 
	once is setting $y$ to 1 while setting $x$ to $x_a+2y_a$, where $x_a,y_a$ are 
	the 
	values of $x,y$ before the start of the phase. Similarly, The total effect of 
	performing Phase (b) once is setting $x$ to $1$ while setting $y$ to $y_b 
	+2x_b$, 
	where $x_b,y_b$ are the values of $x,y$ before the start of the phase. Hence, 
	the 
	total 
	effect of consecutively performing Phases (a) and (b) once can be bounded from 
	below as follows: setting $x$ to 
	$1$ and multiplying $y$ by $4$. Hence, the total effect of performing $N$ 
	consecutive iterations of Phases (a) and (b) is setting $x$ to $1$, multiplying 
	$y$ by $4^N$ and decreasing $z$ by $N$. Since $z$ decreases exactly during the 
	witch from Phase (b) to Phase (a), we can perform exactly $n$ consecutive 
	iterations of (a) and (b). But increasing $y$ from $n$ to $4^n$ requires at 
	least $4^n - n$ steps in VASS, hence the termination complexity of $\A$ is at 
	least exponential. The matching asymptotic upper bound is easy to get.
\end{example}

The key idea of the previous example can be used as building block for showing 
that higher-dimensional terminating VASS satisfying~(D) can have even larger 
termination complexity
than exponential. Already in dimension~4, the complexity can be non-elementary.

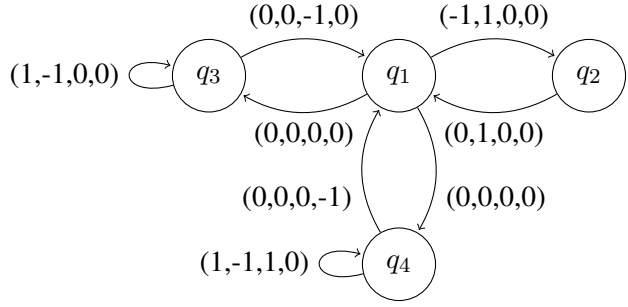
\begin{figure}
	\centering
	\begin{tikzpicture}[shorten >=1pt,node distance=2.5cm,on grid,auto]
	\node[state] (q_1)   {$q_1$};
	\node[state] (q_2) [right=of q_1] {$q_2$};
	\node[state] (q_3) [left=of q_1] {$q_3$};
	\node[state] (q_4) [below=of q_1] {$q_4$};
	\path[->]
	(q_1) edge [bend left] node {(-1,1,0,0)} (q_2)
	edge [bend left] node {(0,0,0,0)} (q_3)
	edge [bend left, pos=0.55] node {(0,0,0,0)} (q_4)
	(q_2) edge [bend left] node  {(0,1,0,0)} (q_1)
	(q_3) edge [bend left] node {(0,0,-1,0)} (q_1)
	edge [loop left] node {(1,-1,0,0)} ()
	(q_4) edge [bend left, pos=0.45] node {(0,0,0,-1)} (q_1)
	edge [loop left] node {(1,-1,1,0)} ()
	;
	\end{tikzpicture}
	\caption{A 4-dimensional VASS satisfying condition (D) which has a 
		non-elementary termination complexity.}
	\label{fig:nonelem}
\end{figure}

\begin{example}
	\label{ex:nonelem}
	Consider the 4-dimensional strongly connected VASS in Fig.~\ref{fig:nonelem}. 
	As before, we denote by $x,y,z,w$ the individual counters. 
	
	For $n\in\N$ we construct a zero-avoiding computation $\alpha(n)$ started in 
	$q_1\vec{n}$ 
	whose length in non-elementary.
	The construction again proceeds by switching 
	between various phases and the phases we consider are the following: in Phase 
	(a) we iterate cycle $q_1,(-1,1,0,0),q_2,(0,1,0,0),q_1$ from $q_1$ as long as 
	$x\geq 2$. The effect of a single execution of (a) is setting $x$ to $1$ and 
	$y$ to $y_a+2x_a$ (as before $v_{p}$ denotes the value of counter $v$ at the 
	start 
	of phase (p)). In Phase (b) we iterate the self-loop on $q_3$ as long as 
	$y\geq 2$, the effect of the phase is setting $y$ to $1$ and $x$ to 
	$x_b+y_b$. 
	Phase~(c) consists of iterating the self-loop on $q_4$ as long as $y\geq 2$ and 
	the effect is setting $y$ to $1$ and $x$ and $z$ to $x_c + y_c$ and 
	$z_c+y_c$, respectively. Switching from (b) to (a) or (c) decreases $z$ by 
	1, 
	while 
	switching from (c) to (a) or (b) decreases $w$ by 1. Now the construction of 
	$\alpha(n)$ proceeds as follows: we switch between Phases (a) and (b) as long 
	as $z\geq 2$, after which we perform Phase (a) once more. We call this a Phase 
	(d) and the total effect of (d) is setting 
	$x$ and $z$ to $1$, and $y$ to a number at least $4^{z_d}\cdot x_d$. After 
	Phase (d) we go to $q_4$ and execute Phase (c), after which we go to $q_1$ and 
	start (d) again, repeating the process until a configuration with a zero 
	counter is hit. The total effect of a single consecutive execution of (d) and 
	(c) is setting $y$ to $1$ and $x$ and $z$ to a number at least $ 
	2^{z_d}\cdot 
	x_d$. Since $w$ is only decremented when switching from (d) to (c), we can 
	repeat this consecutive execution at least $n$ times. An easy induction shows 
	that after $i$ repeats of the consecutive executions of (d) and (c) the value 
	of $x$ is at least $$\xi_n:= 
	n\cdot 2
	\underbrace{
		{{{^{2\vphantom{h}}}^{2\vphantom{h}}}^{\cdots\vphantom{h}}}^{2^n\vphantom{h}}
	}_{\text{$n$ times}}
	. $$
	Hence, the length of $\alpha(n)$ is at least $\xi_n$, i.e. non-elementary.
\end{example}

Figure~\ref{fig:exponential} also provides an example showing that 
lexicographic 
ranking functions are not sound for polynomial bounds on termination 
complexity. We first define the notion of lexicographic ranking function for 
VASS: we specialize the standard definition of a lexicographic ranking 
functions 
for affine automata~\cite{ADFG10:lexicographic-ranking-flowcharts} (a 
generalization of VASS which models general linear 
arithmetic programs). Formally, an $m$-dimensional lexicographic map for a VASS 
$\A=(Q,T)$ is a 
collection $\{f_q^j\mid q \in Q,1\leq j \leq m\}$ of linear functions of 
counter 
values, one function per state and $1\leq j \leq m$ (we allow $m$ to be 
different from the dimension $d$ of $\A$). A lexicographic map $\{f_q^j\mid 
q \in Q,1\leq j \leq m\}$ is a lexicographic $\epsilon$-ranking function for  
$\A$ 
if each 
$f_q^j$ is bounded from below on $\N$ and for 
each 
transition $(q,\bu,q')$ of $\A$ there exists $1\leq j \leq m$ such that 
$f_{q'}^{j}(\bu)\leq f^j_{q}(\vec{0}) - \epsilon$ and for all $1\leq j'<j$ it 
holds 
$f_{q'}^{j'}(\bu)\leq f^{j'}_{q}(\vec{0})$. A standard argument shows that if 
$\A$ has a lexicographic $\epsilon$-ranking function for, then it is 
terminating. 
However, 
lexicographic ranking functions are not sound for polynomial complexity bounds.

\begin{example}
\label{ex:lex}
Consider the VASS $\A$ in Figure~\ref{fig:exponential}. Then, denoting the first, 
second, and third counter as $x,y,z$, respectively, there is the following 
3-dimensional lexicographic $\frac{1}{2}$-ranking function for $\A$ (we denote 
$f_q = 
(f^1_q,\dots,f^m_q)$): $f_{q_1}=(z,y,x)$, $f_{q_2}=(z,y-\frac{1}{2},x)$, 
$f_{q_4}=(z-\frac{1}{2},x,y)$, $f_{q_3}=(z-\frac{1}{2},x-\frac{1}{2},y)$. But 
as shown in Example~\ref{ex:exponential}, the VASS has exponential termination 
complexity.
\end{example}

\section{Technical Proofs}
\label{sec-tech-proofs}

\subsection{Proof of Theorem~\ref{thm-linear-polytime}}

We describe a polynomial time-algorithm for deciding 
whether a given VASS has linear termination complexity.
Recall from the proof sketch that it suffices to solve an equivalent problem whether there is an open half-space $\calH_{\bn}$ of $\R^d$ such that $\bn \geq \vec{0}$ and  $\Inc \subseteq \calH_{\bn}$.
    
Let us formalize our intuition presented in the proof sketch. We need to introduce some additional notation: An~\emph{infinite path} $\pi$ is an infinite sequence of the form $p_0,\bu_1,p_1,\bu_2,p_2,\ldots$ where for each $n\geq 1$ the finite subsequence $p_0,\bu_1,p_1,\bu_2,p_2,\ldots,\bu_n,p_n$ is a finite path. We denote by $\pi_{\downarrow n}$ the finite prefix $p_0,\bu_1,p_1,\bu_2,p_2,\ldots,\bu_n,p_n$ of $\pi$.
    Given an infinite path $\pi$, we define the \emph{mean change of $\bn$-value} as
    \[
    \mathit{MC}_{\bn}(\pi) = \liminf_{n\rightarrow \infty} \frac{\eff(\pi_{\downarrow n})\cdot \bn}{n}.
    \]
    Consider the following linear program $\mathcal{L}$ obtained 
    from~\cite{Puterman:1994}, Section~8.8, by substituting the reward $r(s, 
    a)$ with $\bu\cdot \bn$ where $\bu$ is an effect of a transition: 
    
    {\itshape {\bfseries Minimize} $g$ with respect to the following constraints:
    
    For all $(q, \bu, q')\in T$
    \[
    g + h(q) - h(q') \geq \bu\cdot \bn
    \]
    and
    \[
    \bn\geq \vec{0}.
    \]}
    Here, the variables are $g$, all $h(q)$, $q\in Q$, and all components of $\bn$. By applying the results of \cite{Puterman:1994}, for every optimal solution $g, h, \bn$ we have that 
    \[
    	g = \sup_{\pi} \mathit{MC}_{\bn}(\pi).
    \]
    Moreover, there is at least one feasible solution.
   
    We prove that there is $\bn\geq \vec{0}$ such that the open half-space $\calH_{\bn}$ contains $\Inc$ iff an optimal solution $g, h, \bn$ of the above program satisfies $g<0$.
    
    Consider an optimal solution $g, h, \bn$ of the above program. Assume that $g<0$. We show that $\Inc \subseteq \calH_{\bn}$. 
    For the sake of contradiction, assume that there is a short cycle $\pi$ such that $\eff(\pi)\cdot \bn\geq 0$. Following the cycle $\pi$ ad infinitum determines an infinite path $\pi$ with $\mathit{MC}_{\bn}(\pi)\geq 0$. However, this contradicts the fact that
    $0 > g = \sup_{\pi} \mathit{MC}_{\bn}(\pi)$.
    
    Now assume there is $\bn\geq \vec{0}$ such that $\Inc \subseteq  \calH_{\bn}$. Let $\pi$ be an infinite path. Let us fix $n\geq 1$ and consider $\Decomp(\pi_{\downarrow n})$, the decomposition of $\pi_{\downarrow n}$ into short cycles.
    Let $\mathit{Rest}(\pi_{\downarrow n})$ be the remaining path obtained after removing all short cycles of $\Decomp(\pi_{\downarrow n})$ from $\pi_{\downarrow n}$. Note that the length of $\mathit{Rest}(\pi_{\downarrow n})$ is at most $|Q|$, and hence $\eff(\mathit{Rest}(\pi_{\downarrow n}))\cdot \bn\leq \vec{|Q|}\cdot \bn$. 
    
    Now let $m$ be the length (i.e., the number of elements) of the list $\Decomp(\pi_{\downarrow n})$. Note that \mbox{$m\geq n/|Q|-1$}.
    Consider $\varepsilon>0$ such that for all short cycles $\alpha$ we have that $\eff(\alpha)\cdot \bn\leq -\varepsilon$. Then 
    \[
    \eff(\pi_{\downarrow n})\quad \leq\quad m\cdot (-\varepsilon) + \vec{|Q|}\cdot \bn\quad \leq\quad (n/|Q|-1)\cdot (-\varepsilon) + \vec{|Q|}\cdot \bn \quad = \quad (n\cdot (-\varepsilon)/|Q|) +(\vec{|Q|}\cdot \bn +\varepsilon)
    \]
    and thus 
    \[
    \frac{\eff(\pi_{\downarrow n})}{n}\quad\leq \quad \frac{(n\cdot (-\varepsilon)/|Q|) +(\vec{|Q|}\cdot \bn +\varepsilon)}{n} \quad =\quad
    \frac{-\varepsilon}{|Q|} + \frac{(\vec{|Q|}\cdot \bn +\varepsilon)}{n}.
    \]
    Since $\lim_{n \rightarrow \infty}(\vec{|Q|}\cdot \bn +\varepsilon)/n = 0$, we obtain that $\mathit{MC}_{\bn}(\pi) \leq (-\varepsilon)/|Q|$. As $\pi$ was chosen arbitrarily, we have that
    \[
    \sup_{\pi} \mathit{MC}_{\bn}(\pi)\quad  \leq \quad \frac{-\varepsilon}{|Q|}\quad <\quad 0.
    \] 
    Hence, there is a solution $g, h, \bn$ of the above linear program with $g<0$.
    
    In order to decide whether there is an open half-space $\calH_{\bn}$ of 
    $\R^d$ such that $\bn \geq \vec{0}$ and  $\Inc \subseteq \calH_{\bn}$, it 
    suffices to compute an optimal solution $g, h, \bn$ of the above linear 
    program, which can be done in polynomial time (see, e.g., 
    \cite{Khachiyan:1979}), and check whether $g<0$. 
    
    Now we get back to weighted linear ranking functions. Note that each 
    solution 
    $g,h,\bn$ of the linear program $\mathcal{L}$ in which $g<0 $ yields a 
    weighted linear
    ranking 
    function $(\bc,\{h_q\mid q\in Q\})$ by putting $\bc:=\bn$ and $h_q:=h(q)$ 
    for each $q$. Conversely, each weighted linear ranking function yields a 
    solution of $\mathcal{L}$ where $g<0$, (we need to put $g:=-\epsilon$, 
    where $\epsilon$ is from the definition of a weighted lin. ranking 
    function). Hence, a 
    VASS $\A$
    has 
    linear termination complexity if and only if it has a weighted linear 
    ranking function and this can be decided in polynomial time in size of $\A$.
    \qed

\subsection{Proof of Theorem~\ref{thm-condAB}}
\label{sec-condAB}

If condition (A) or (B) holds, there is no $\bn \geq \vec{0}$ such that $\Inc \subseteq \hat{\calH}_{\bn}$. We show that then there exists $\bu \in \cone{}(\Inc)$ such that $\bu > \vec{0}$. Suppose there is no such $\bu$. Let $B$ be the set of all $\bv > \vec{0}$. Since $\cone{}(\Inc)$ and $B$ are convex and disjoint, there is a separating hyperplane with normal $\bn \geq \vec{0}$ for $\cone{}(\Inc)$ and $B$. Since $\cone{}(\Inc) \subseteq \hat{\calH}_{\bn}$, we have a contradiction.

So, let $\bu > \vec{0}$ such that $\bu = \sum_{i=1}^k a_i \cdot \bv_i$, where $a_i \in \Q^+$ and $\bv_i \in \Inc$ for all $1 \leq i \leq k$. Hence, there also exist $b_1,\ldots,b_k \in \N^+$ such that $\bw =  \sum_{i=1}^k b_i \cdot \bv_i > \vec{0}$. Let us fix a cycle $\pi$ in $\A$ visiting all control states (here we need that $\A$ is strongly connected). Clearly, there exists $c \in \N$ such that $\eff(\pi) + c\cdot \bw > 0$. Let $\varrho$ be a cycle obtained from $\pi$ by inserting $c\cdot b_i$ copies of a short cycle $\gamma_i$, where $\eff(\gamma_i) = \bv_i$. Then, $\eff(\varrho) > \vec{0}$, and hence there exists an infinite computation initiated in $p\vec{n}$ for a sufficiently large $n \in \N$ (the control state $p$ can be chosen arbitrarily). 

\subsection{Proof of Lemma~\ref{lem-good-normal}}
\label{sec-good-normal}

Due to condition~(C), there exists at least one positive normal. Hence, we can fix a positive $\bn \in \Normals(\A)$ such that the set $\{\bu \in \Inc \mid \bu \cdot \bn = 0 \}$ is \emph{minimal}. We show that for every $\bv \in \cone{}(\Inc)$ we have that $-\bv \in \cone{}(\Inc)$ iff $\bv \cdot \bn = 0$, i.e., $\bn$ is a good normal. The ``$\Rightarrow$'' direction immediate---if $\bv,-\bv \in \cone{}(\Inc)$, then $\bv \cdot \bn \leq 0$ and $-\bv \cdot \bn \leq 0$, which implies $\bv \cdot \bn = 0$. For the other direction, suppose there exists $\bv \in \cone{}(\Inc)$ such that $\bv \cdot \bn = 0$ and $-\bv \not\in \cone{}(\Inc)$. Then there also exists $\bu \in \Inc$ such that $\bu \cdot \bn = 0$ and $-\bu \not\in \cone{}(\Inc)$. For the rest of this proof, we fix such $\bu$. By Farkas' lemma, there exists a separating hyperplane for $\cone{}(\Inc)$ and $-\bu$ with normal vector $\bn'$, i.e., $-\bu \cdot \bn' > 0$ and $\bv \cdot \bn' \leq 0$ for every $\bv \in \cone{}(\Inc)$. Let us fix a sufficiently small $\varepsilon > 0$ such that $\bn + \varepsilon \bn' > \vec{0}$ and $\bv \cdot (\bn + \varepsilon \bn') < 0$ for all $\bv \in \Inc$ where $\bv \cdot \bn < 0$. Clearly, $\bn + \varepsilon \bn'$ is a positive normal. Further, for all $\bv \in \cone{}(\Inc)$ such that $\bv \cdot \bn < 0$ we have that $\bv \cdot (\bn + \varepsilon \bn') < 0$. Since $\bu \cdot (\bn + \varepsilon \bn') < 0$, we obtain a contradiction with the minimality of $\bn$.	

To compute a good normal, first observe that the condition of Definition~\ref{def-good-normal} can be safely relaxed just to the vectors of $\Inc$, i.e., if $\bn \in\Normals(\A)$ such that 
$\bn > \vec{0}$ and $-\bv\in \cone{}(\Inc)$ iff $\bv\cdot \bn=0$ for every $\bv\in \Inc$, then $\bn$ is a good normal. To see this, fix some $\bn$ with this property, and let $\bu = \sum_{i=1}^k a_i\cdot \bv_i$, where $a_i \in \R^+$ and $\bv_i \in \Inc$ for all $1 \leq i \leq k$. We need to show that $-\bu\in \cone{}(\Inc)$ iff $\bu \cdot \bn=0$. If $-\bu \in \cone{}(\Inc)$, then $-\bu = \sum_{i = 1}^{k'} a'_j \cdot \bv'_i$ where $a'_i \in \R^+$ and $\bv'_i \in \Inc$ for all $1 \leq i \leq k'$. Hence, 
\[
\vec{0} = \bu+ (-\bu) = \sum_{i =1}^k a_i\cdot \bv_i + \sum_{i =1}^{k'} a'_i \cdot \bv'_i .
\]
Hence,
\[
0 = (\bu+ (-\bu)) \cdot \bn = \sum_{i =1}^k a_i\cdot \bv_i \cdot \bn + \sum_{i =1}^{k'} a'_i \cdot \bv'_i \cdot \bn.
\]
Since $\bv_i \cdot \bn \leq 0$ and $\bv'_i \cdot \bn \leq 0$ for all $1 \leq i \leq k$ and all $1 \leq i' \leq k'$, we obtain $\bv_i \cdot \bn=0$ for all $1 \leq i \leq k$, hence $\bu\cdot \bn=0$. On the other hand, if $\bu \cdot \bn=0$, then $\sum_{i =1}^k a_i \cdot \bv_i \cdot \bn = 0$. Since  $a_i>0$ and $\bv_i\cdot \bn\leq 0$  for all $1 \leq i \leq k$, we have that $\bv_i \cdot \bn=0$ for all $1 \leq i \leq k$. Hence $-\bv_i\in \cone{}(\Inc)$ for every $1 \leq i \leq k$ (by our assumption), and $-\bu = \sum_{i=1}^k a_i \cdot (-\bv_i) \in \cone{}(\Inc)$. 

Using the above observation, we can compute a good normal using linear programming as follows: First, compute the set $I = \{\bv \in \Inc \mid -\bv \in \cone{}(\Inc)\}$. Note that $I$ can be computed easily by checking feasibility of the following linear constraints:
\[
-\bv = \sum_{\bu \in \Inc} a_{\bu} \cdot \bu \qquad \text{ and } \qquad a_{\bu} \geq 0.
\]
Here, the variables are $a_{\bu}$. A good normal can be computed using the following linear program:

{\itshape
	{\bfseries Maximize} $\varepsilon$ with respect to the following constraints:
\begin{gather*}
	\bu \cdot \bn = 0\text{ for all } \bu \in I\\
	\bv \cdot \bn \leq -\varepsilon \text{ for all } \bv \in \Inc \smallsetminus I \\
	\bn\geq \vec{\varepsilon}.
\end{gather*}
}
Here, the variables are $\varepsilon$ and all components of $\bn$.

Note that there is a good normal iff there is an optimal solution with $\varepsilon>0$. Moreover, every optimal solution $\varepsilon, \bn$ with $\varepsilon>0$ gives a good normal $\bn$.

\subsection{Proof of Theorem~\ref{thm-A-decomp}}
\label{sec-A-decomp}

Now start by formulating an auxiliary technical lemma which is needed in the proof of Theorem~\ref{thm-A-decomp}.

\begin{lemma}
	\label{lem-kappa}
	Let $\A$ be a VASS satisfying~(C), and let $\bn$ be a good normal. Then there is a constant $\kappa \in \R^+$ such that for every $\bw \in \cone{}(\Inc)$, where $\bw \cdot \bn = 0$ and $\Norm(\bw) = 1$, there exist $k \in \N$, $a_1,\ldots,a_k \in \R^+$, and $\bv_1,\ldots,\bv_k \in \Inc$ such that $\bw = \sum_{j=1}^k a_j \cdot \bv_j$, $\bv_j \cdot \bn = 0$ for all $1 \leq j \leq k$, and for all $k' \leq k$, the absolute values of all components of the vector $\sum_{j=1}^{k'} a_j \cdot \bv_j$ are bounded by~$\kappa$. 
\end{lemma}
\begin{proof}
	Let $\bw = \sum_{j=1}^k a_j \cdot \bv_j$ where $a_j \in \R^+$, $\bv_j \in \Inc$ for all $1 \leq j \leq k$, and $k$ is \emph{minimal}. Clearly, $\bw  \cdot \bn = \sum_{j=1}^k a_j \cdot (\bv_j \cdot \bn) = 0$, which implies $\bv_j \cdot \bn = 0$ for all $1 \leq j \leq k$ (recall that $\bv_j \cdot \bn \leq 0$ because $\bn \in \Normals(\A)$). First, we show that for every $j \leq k$, the vector $-\bv_j$ does not belong to $\cone{}(\{\bv_1,\ldots,\bv_{j-1},\bv_{j+1},\ldots,\bv_k\})$. Assume the converse, i.e., $-\bv_1 \in \cone{}(\{\bv_2,\ldots,\bv_k\})$. Then $-\bv_1 = \sum_{j=2}^k b_j \cdot \bv_j$, where $b_j \in \R^+$ for all $2 \leq j \leq k$. Further, 
	\[
	\bw \quad = \quad (a_1-c) \cdot \bv_1 
	\ + \ (a_2 - c b_2)\cdot \bv_2 \ + \ \cdots \ + \ (a_k -c b_k)\bv_k
	\]
	for every $c > 0$. Clearly, there exists $c>0$ such that at least one of the coefficients $(a_1-c)$, $(a_2 - c b_2),\ldots, (a_k -c b_k)$ is zero and the other remain positive, which contradicts the minimality of~$k$. Since $\{\bv_1,\ldots,\bv_k\} \subseteq \hat{\calH}_\bn$, there must exist $\bn' > \vec{0}$ such that $\{\bv_1,\ldots,\bv_k\} \subseteq \calH_{\bn'}$ (otherwise, we can use the same argument as in the proof of Case~(a) of Lemma~\ref{lem-no-halfspace} to show that \mbox{$-\bv_j \in \cone{}(\{\bv_1,\ldots,\bv_{j-1},\bv_{j+1},\ldots,\bv_k\})$} for some $1 \leq j \leq k$). Since $\bv_j \cdot \bn' < 0$ for all $1 \leq j \leq k$, each $\bv_j$ moves in the direction of $-\bn$ by some fixed positive distance. Since $\Norm(\bw) =1$, there is a bound $\delta_{\bv_1,\ldots,\bv_k} \in \R^+$ such that $a_j \leq \delta_{\bv_1,\ldots,\bv_k}$ for all $1 \leq j \leq k$, because no $a_j\cdot \bv_j$ can go in the direction of $-\bn$ by more than a unit distance. 
	
	The above claim applies to every $\bw \in \cone{}(\Inc)$ where $\bw \cdot \bn = 0$. Since $\Inc$ is finite, there are only finitely many candidates for the set of vectors $\{\bv_1,\ldots,\bv_k\}$ used to express $\bw$, and hence there exists a fixed upper bound $\delta \in \R^+$ for all $\delta_{\bv_1,\ldots,\bv_k}$. This means that, for every $\bw \in \cone{}(\Inc)$ where $\bw \cdot \bn = 0$, there exist $k \in \N$, $a_1,\ldots,a_k \in \R^+$, and $\bv_1,\ldots,\bv_k \in \Inc$ such that $\bw = \sum_{j=1}^k a_j \cdot \bv_j$, $\bv_j \cdot \bn = 0$, and $a_j \leq \delta$ for all $1 \leq j \leq k$. This immediately implies the existence of~$\kappa$.
\end{proof}

Now can formalize the proof of Theorem~\ref{thm-A-decomp}.

 \emph{All cycles of $\A^{\bn}_{C_i}$ are $\bn$-neutral.} \quad First, realize that for every cycle $\eta$ of $\A$ (not necessarily short) we have that $\eff(\eta) \cdot \bn = \sum_{\gamma \in \Decomp(\eta)} \eff(\gamma) \cdot \bn \leq 0$. Now let
 $\beta = p_0,\bu_1,p_1,\bu_2,p_2,\ldots,\bu_n,p_k$ be a cycle of $\A^{\bn}_{C_i}$ (not necessarily short). Then each transition $(p_j,\bu_{j+1},p_{j+1})$ of $\beta$ is contained in some $\bn$-neutral short cycle $\gamma_i$ of $\A$. Let $\varrho_i$ be the (unique) path from $p_{j+1}$ to $p_j$ determined by $\gamma_j$, and let $\varrho = \varrho_{k-1} \odot \cdots \odot \varrho_0$. Then $\eff(\beta) + \eff(\varrho) = \sum_{j=0}^{k-1} \eff(\gamma_j)$. Hence, 
 $\eff(\beta)\cdot \bn + \eff(\varrho) \cdot \bn = \sum_{j=0}^{k-1} \eff(\gamma_j) \cdot \bn = 0$. Thus, we obtain $\eff(\beta)\cdot \bn = - \eff(\varrho) \cdot \bn$. Since both $\beta$ and $\varrho$ are cycles of $\A$, we have that $\eff(\beta)\cdot \bn \leq 0$ and $\eff(\varrho) \cdot \bn \leq 0$, which implies $\eff(\beta)\cdot \bn = 0$.
 
 \emph{Constructing the paths of length $\Theta(f_i(n))$.}\quad Since the termination complexity of $\A^{\bn}_{C_i}$ is $\Theta(f_i(n))$, there is $b \in \R^+$ such that for all sufficiently large $n \in \N$ there exist a configuration $p_n \vec{n}$ and a zero-avoiding computation $\beta_n$ of length at least $b \cdot f_i(n)$ initiated in $p_n \vec{n}$. Since $\pi_{\beta_n}$ inevitably contains a cycle whose length is at least $b' \cdot f_i(n)$ (for some fixed $b' \in \R^+$ independent of $\beta_n$), we can safely assume that $\pi_{\beta_n}$ is actually a cycle, which implies $\eff(\pi_{\beta_n}) \in \cone{}(\Inc)$.
 
 \emph{Constructing the compensating path.}\quad Since $\pi_{\beta_n}$ is $\bn$-neutral and $\eff(\pi_{\beta_n}) \in \cone{}(\Inc)$, we have that $-\eff(\pi_{\beta_n}) \in \cone{}(\Inc)$. This is where we use the defining property of a good normal. Since $-\eff(\pi_{\beta_n}) = \sum_{j=1}^m a_j \cdot \bv_j$, where $m \in \N$, $a_j \in \Q^+$, and $\bv_j \in \Inc$ for all $1 \leq j \leq m$, a straightforward idea is to define the compensating path by ``concatenating'' $\lfloor a_j \rfloor$ copies of $\gamma_j$, where $\eff(\gamma_j) = \bv_j$, for all $1 \leq j \leq m$. This would produce the desired effect on the counters, but there is no bound on the counter decrease in intermediate configurations visited when executing this path. To overcome this problem, we construct the compensating path for $\pi_{\beta_n}$ more carefully. Let $\bw$ be the normalized $\eff(\pi_{\beta_n})$, i.e., $\bw$ has the same direction as $\eff(\pi_{\beta_n})$ but its norm is equal to~$1$. By Lemma~\ref{lem-kappa}, $-\bw$ is expressible as $-\bw = \sum_{j=1}^m a_j \cdot \bv_j$, where $m \in \N$, $a_j \in \Q^+$, and $\bv_j \in \Inc$, so that $\bv_j \cdot \bn = 0$ for all $1\leq j \leq m$, and for all $m' \leq m$, the  absolute values of all components of the vector $\sum_{j=1}^{m'} a_j \cdot \bv_j$ are bounded by~$\kappa$, where $\kappa$ is a constant independent of $\bw$. Let us fix some cycle $\eta$ of $\A^{\bn}_{C_i}$ visiting all of its states (recall that $\A^{\bn}_{C_i}$ is strongly connected). The compensating path for $\pi_{\beta_n}$ is obtained from $\eta$ by inserting $\lfloor \Norm(\eff(\pi_{\beta_n})) \cdot a_j \rfloor$ copies of a short cycle with effect $\bv_j$, for every $1 \leq j \leq m$. Observe that the difference between the effect of this compensating path and $-\eff(\pi_{\beta_n})$ is bounded by a constant vector independent of~$n$. Further, when executing the compensating path, the counters are never decreased by more that $\kappa \cdot \Norm(\eff(\pi_{\beta_n}))$.
 
 \emph{Constructing a zero-avoiding computation $\alpha_n$ of length $\Omega(n\cdot f_i(n))$.}\quad Now we are ready to put the above ingredients together, which still requires some effort. Let us fix a sufficiently large $n \in \N$ and a configuration $p\bv$ where $\size{p\bv} = n$ and $p$ is a control state of $\A^{\bn}_{C_i}$. Let $q$ be the first state of $\pi_{\beta_n}$. If we started $\alpha_n$ in $p\bv$ by executing a finite path which changes the control state from $p$ to the first control state of $\pi_{\beta_n}$ (which takes at most $|Q|$ transitions) and continued by executing $\pi_{\beta_n}$, the counters could potentially reach values arbitrarily close to zero (it might even happen that $\pi_{\beta_n}$ is not executable). Instead, we fix a suitable $n' \leq n$ satisfying $n - n' \geq \kappa \cdot \Norm(\eff(\pi_{\beta_{n'}})) + |Q|$. Since $\Norm(\eff(\pi_{\beta_{n'}})) \leq \sqrt{d} \cdot n'$, we can safely put $n' = (n-|Q|)/(1+\kappa \sqrt{d})$. Now, we can initiate $\alpha_n$ by a short finite path which changes the control state from $p$ to the first control state of $\pi_{\beta_{n'}}$, and continue by executing $\pi_{\beta_{n'}}$. Note that $(1+\kappa \sqrt{d})$ is a constant, so decreasing $n$ to $n'$ has no influence in the asymptotic length of the constructed computation. 
 Then, we can safely execute the compensating path for $\pi_{\beta_{n'}}$, and thus reach a configuration $q\bu$ where we continue in the same way as in $p\bv$, i.e., execute another finite path of length $\Theta(f_i(n))$ and its corresponding compensating path. Since the $\bv-\bu$ is bounded by a constant vector, this can be repeated $\Omega(n)$ times before reaching a configuration where some counter value is not sufficiently large to perform another ``round''. Hence, the length of the resulting $\alpha_n$ is $\Omega(n\cdot f_i(n))$.

\section{Related Work}
In this section we discuss the related work. 

\smallskip\noindent{\em Resource analysis.}
Our work is most closely related to automatic amortized 
analysis~\cite{DBLP:conf/aplas/HoffmannH10,DBLP:conf/esop/HoffmannH10,DBLP:conf/popl/HofmannJ03,DBLP:conf/esop/HofmannJ06,DBLP:conf/csl/HofmannR09,DBLP:conf/fm/JostLHSH09,DBLP:conf/popl/JostHLH10,Hoffman1,DBLP:conf/popl/GimenezM16},
as well as the SPEED project~\cite{SPEED1,SPEED2,DBLP:conf/cav/GulavaniG08}.
All these works focus on worst-case asymptotic bounds for programs,
and present sound methods but not complete methods for upper bounds, 
i.e., even though the asymptotic bound is linear or quadratic, the approaches 
may still fail to provide any upper bound.
However, all these works consider general programs rather than the model of 
VASS. 
In contrast, we consider VASS and present sound and complete method to 
derive tight (upper and matching lower) polynomial complexity bounds.

\smallskip\noindent{\em Recurrence relations.}
Other approaches for bounds analysis involve recurrence relations, such as 
~\cite{DBLP:conf/icfp/Grobauer01,DBLP:journals/tcs/FlajoletSZ91,DBLP:journals/entcs/AlbertAGGPRRZ09,DBLP:conf/sas/AlbertAGP08,DBLP:conf/esop/AlbertAGPZ07}.
Even for relatively simple programs the recurrence are quite complex, and cannot be obtained 
automatically.
In contrast, we present a polynomial-time approach for optimal asymptotic bounds for VASS.

\smallskip\noindent{\em Ranking functions and extensions.}
Ranking functions for intraprocedural analysis have been widely 
studied~\cite{BG05,DBLP:conf/cav/BradleyMS05,DBLP:conf/tacas/ColonS01,DBLP:conf/vmcai/PodelskiR04,DBLP:conf/pods/SohnG91,DBLP:conf/vmcai/Cousot05,DBLP:journals/fcsc/YangZZX10,DBLP:journals/jossac/ShenWYZ13}.
Most works have focussed on linear or polynomial ranking functions~\cite{DBLP:conf/tacas/ColonS01,DBLP:conf/vmcai/PodelskiR04,DBLP:conf/pods/SohnG91,DBLP:conf/vmcai/Cousot05,DBLP:journals/fcsc/YangZZX10,DBLP:journals/jossac/ShenWYZ13},
as well as non-polynomial bounds~\cite{CFG17}.
Again, these approaches are sound, but not complete even to derive upper bounds 
for VASS.
The notion of ranking functions have been also extended to ranking 
supermartingales~\cite{SriramCAV, HolgerPOPL, DBLP:conf/popl/ChatterjeeFNH16, CFG16, CNZ17} 
for expected termination time of probabilistic programs, but such approaches do not present 
polynomial asymptotic bounds.

\smallskip\noindent{\em Results on VASS.} 
The model of VASS~\cite{KM69} or equivalently Petri nets are a fundamental model 
for parallel programs~\cite{EN94,KM69} as well as parameterized systems~\cite{Bloem16}.
The termination problems (counter-termination, control-state termination) as well
as the related problems of boundedness and coverability have been a rich source 
of theoretical problems
that have been widely 
studied~\cite{Lipton:PN-Reachability,Rackoff:Covering-TCS,Esparza:PN,ELMMN14:SMT-coverability,BG11:Vass}.
The complexity of the termination problem with fixed initial configuration is 
\EXPSPACE-complete~\cite{Lipton:PN-Reachability,Yen92:Petri-Net-logic,AH11:Yen}.
Recent work such as~\cite{SZV14,Bloem16} shows how VASS and subclass of VASS (such as lossy VASS)
provide a natural model for abstraction and analysis of programs as well as 
parametrized systems. 
The work of~\cite{SZV14} also considers lexicographic ranking functions to obtain sound 
asymptotic upper bounds for lossy VASS. 
However, this approach is not complete, and also do not consider tight complexity bounds 
(but only upper bounds).
Besides the termination problem, the more general reachability problem where given a VASS, 
an initial and a final configuration, whether there exists a path between them has also been 
studied~\cite{Mayr:PN-reachability,Kosaraju82:VASS-reach-dec,Leroux:VASS-POPL}.
The reachability problem is 
decidable~\cite{Mayr:PN-reachability,Kosaraju82:VASS-reach-dec,Leroux:VASS-POPL},
 and 
\EXPSPACE-hard~\cite{Lipton:PN-Reachability}, 
and the current best-known upper bound is cubic 
Ackermannian~\cite{LS15:VASS-reach-complexity}, a complexity class belonging to 
the third level of a fast-growing complexity hierarchy introduced 
in~\cite{Schmitz16:hyperackermannian-complexity-hierarchy}.

Other related approaches are sized types~\cite{DBLP:journals/lisp/ChinK01,DBLP:conf/icfp/HughesP99,DBLP:conf/popl/HughesPS96}, and polynomial resource bounds~\cite{DBLP:conf/tlca/ShkaravskaKE07}.
Again none of these approaches are complete for VASS nor they can yield tight asymptotic complexity bounds.

\smallskip\noindent{\em Hyperplane-separation technique and existence of infinite computation.}
The problem of existence of infinite computations in VASS has been studied in 
the literature.
Polynomial-time algorithms have been presented in~\cite{CDHR10,VCDHRR15} using results of~\cite{KS88}.
In the more general context of games played on VASS, even deciding the existence of 
infinite computation is coNP-complete~\cite{CDHR10,VCDHRR15}, and various algorithmic approaches 
based on hyperplane-separation technique have been studied~\cite{CV13,JLS15,CJLS17}.
In this work we also consider normals of effects of cycles in VASS, which is related to 
hyperplane-separation technique.
However all previous works consider hyperplane-based techniques to determine the 
existence of infinite computations on games played on VASS, and do not consider
asymptotic time of termination.
In contrast, we present the first approach to show that hyperplane-based techniques 
can be used to derive tight asymptotic complexity bounds on termination time for VASS.

\section{Conclusion}
\label{sec-concl}
In this paper, we studied the problem of obtaining precise polynomial asymptotic bounds for
VASS. We obtained a full end efficient characterization of all VASS with linear termination complexity. Then we considered polynomial termination for strongly connected VASS, dividing them into four disjoint classes (A)--(D). For the first two classes, we proved that the VASS are non-terminating. For VASS in~(C), we obtained a full and effective characterization of termination complexity. For the last class~(D), we have shown that the termination complexity can be exponential even for dimension three. The results are applicable also to general (i.e., non-strongly connected VASS), by analyzing the individual SCCs. Some extra effort is needed in~(C), because here a possible increase in the size of configurations accumulated in a given SCC before moving into another SCC must be taken into account. To keep our proofs reasonably simple, we considered just strongly connected VASS.

Our result gives rise to a number of interesting directions for future work.
First, whether our precise complexity analysis or the complete method can be extended to other
models in program analysis (such as affine programs with loops) is an interesting theoretical
direction to pursue. 
Second, in the practical direction, using our result for developing a scalable tool for sound and 
complete analysis of asymptotic bounds for VASS and their applications in program analysis is 
also an interesting subject for future work.

\bibliographystyle{plain}
\bibliography{tomas,PL,NewCite}

\end{document}